\documentclass{article}
\usepackage{ijcai22}

\usepackage{times}
\usepackage{soul}
\usepackage{url}
\usepackage[hidelinks]{hyperref}
\usepackage[utf8]{inputenc}
\usepackage[small]{caption}
\usepackage{graphicx}
\usepackage{amsmath}
\usepackage{amsthm}
\usepackage{booktabs}
\usepackage{textcomp}
\usepackage{xcolor}
\usepackage{url}
\usepackage[T1]{fontenc}
\usepackage[utf8]{inputenc}
\usepackage{blindtext}
\usepackage{pdfpages}

\usepackage{here}
\usepackage{xcolor}
\usepackage{fancyhdr}
\usepackage[title,titletoc]{appendix}

\usepackage{bm,enumerate,amsmath,amssymb,amsthm}
\newtheoremstyle{assumption}
  {\topsep}
  {\topsep}
  {}
  {1em}
  {\itshape}
  {}
  {.5em}
  {\thmname{#1}\thmnumber{ #2}\thmnote{ (#3)}}
\theoremstyle{assumption}

\newtheorem{lemma}{Lemma}

\newtheorem{theorem}{Theorem}
\newtheorem{example}{Example}

\usepackage{algorithm}
\usepackage[noend]{algpseudocode}

\newcommand{\mc}{\mathcal}

\newcommand{\argmax}{\mathop{\rm argmax~}\limits}

\newcommand{\prev}{\gamma}

\newcommand{\tp}[2]{(#1, #2)}
\newcommand{\ik}{{\tp{i}{k}}}

\newcommand{\newcolor}{black}

\newcommand{\modified}[1]{\textcolor{black}{#1}}
\newcommand{\ijcai}[1]{\textcolor{black}{#1}}

\def\arxiv{1}
\newcommand{\cameracolor}{black}

\if\arxiv0
\newcommand{\camera}[1]{\textcolor{\cameracolor}{#1}}
\else
\newcommand{\camera}[1]{\textcolor{\cameracolor}{#1}}
\fi
\newcommand{\urlblue}[1]{\textcolor{blue}{#1}}

\title{A Polynomial-time Decentralised Algorithm for Coordinated Management of Multiple Intersections}

\author{
    Author Name
    \affiliations
    Affiliation
    \emails
    pcchair@ijcai-22.org
}

\def\anonym{0}
\if\anonym1
\author {
    Paper ID \# 5633
}
\else
\author{
Tatsuya Iwase$^1$
\and
Sebastian Stein$^2$
\and
Enrico H. Gerding$^2$
\and
Archie Chapman$^3$
\affiliations
$^1$Toyota Motor Europe NV/SA, Zaventem, Belgium (on leave from Toyota CRDL Inc., Japan)\\
$^2$University of Southampton, Southampton, United Kingdom\\
$^3$The University of Queensland, Brisbane, Australia\\
\emails
tiwase@mosk.tytlabs.co.jp,
\{ss2, eg\}@ecs.soton.ac.uk,
archie.chapman@uq.edu.au
}
\fi

\begin{document}

\maketitle

\if\arxiv1
\thispagestyle{fancy}
\pagestyle{fancy}
\fancyhf{}
\renewcommand{\headrulewidth}{0pt}
\lfoot[Copyright International
Joint Conferences on Artificial Intelligence (\href{https://ijcai.org}{IJCAI}), 2022. All rights
reserved.]{Copyright International
Joint Conferences on Artificial Intelligence (\href{https://ijcai.org}{\urlblue{IJCAI}}), 2022. All rights
reserved.}

\fi

\color{\newcolor}
\begin{abstract}
Autonomous intersection management has the potential to reduce road traffic congestion and energy consumption. 
To realize this potential, efficient algorithms are needed. 
However, most existing studies locally optimize one intersection at a time, and this can cause negative externalities on the traffic network as a whole.
Here, we focus on coordinating multiple intersections,
and formulate the problem as a distributed constraint optimisation problem (DCOP). We consider three utility design approaches that trade off efficiency and fairness.
Our polynomial-time algorithm for coordinating multiple intersections reduces the traffic delay by about 41\% compared to independent single intersection management approaches.
\if\arxiv0
\renewcommand{\thefootnote}{}
\footnote[0]{\camera{An extended version including technical appendices is available at \href{https://arxiv.org/abs/2105.11292}{https://arxiv.org/abs/2105.11292}.}}
\renewcommand{\thefootnote}{\arabic{footnote}}
\fi
\end{abstract}

\section{Introduction}
\label{sec:intro}

Autonomous intersection management (AIM) is the intelligent coordination of autonomous vehicles at intersections  \cite{dresner2008multiagent}. Here, an intersection manager agent allocates slots in advance to vehicles to minimize the total travel time. AIM has the potential to reduce traffic congestion, the number of accidents, fuel consumption as well as ${\rm CO}_2$ emissions \cite{namazi2019intelligent} compared to using traditional traffic signals.  However, many studies of AIM focus on the management of a single intersection and do not cover the coordination of traffic flows across the entire road network  \cite{namazi2019intelligent}. Indeed, the local optimization of one intersection can cause negative externalities on the surrounding intersections. For example, a vehicle passing an intersection can conflict with a platoon of vehicles at a downstream intersection and cause a cascade of braking. 
Hence, the coordinated management of multiple intersections (CMMI) is required to solve this issue. 

%
However, a key challenge in optimizing traffic over multiple intersections is its inherent computational complexity. \ijcai{As we will discuss later, CMMI contains an NP-hard packing problem.} Due to this, centralized approaches such as a mixed integer linear program (MILP) \cite{bredstrom2008combined}, microscopic models \cite{li2020autonomous}, queue models \cite{wu2012cooperative} and agent-based simulations \cite{jin2012multi,wang2020cooperative} solve the problem for each intersection independently or sacrifice scalability.

\ijcai{Decentralized solutions are generally more scalable, robust without having a single point of failure, and preferable from a privacy perspective, as they do not require sharing  the route information of all cars. They are also suitable for the dynamic nature of CMMI, where vehicles newly appear or change their routes, because agents only have to check the local changes, while a centralized controller has to monitor the entire system. However, a naive application of Distributed Constraint Optimization (DCOP) \cite{fioretto2018distributed}
will again result in controlling each individual intersection
independently \cite{vu2020decentralised},
leading to the same issue of ignoring externality mentioned above. Also, DCOP algorithms such as Max-sum \cite{farinelli2008decentralised} or MGM-k \cite{pearce2005local} are not scalable for multiple intersections due to their computational and communication complexity \cite{chapman2011benchmarking}. DCOP games \cite{chapman2011benchmarking} are a scalable and suitable approach for large problems. However, their application to CMMI is not trivial because the feasibility of the solution is not guaranteed due to the presence of hard constraints that are typically solved in mathematical programming problems using Lagrange multipliers.}

\ijcai{To address the issues above, we formulate a new CMMI game, which is an extension of DCOP games.}
We cast CMMI as a multiagent system (MAS) where agents need to be coordinated without global knowledge of the entire network or excessive communication between agents. Since the performance of a MAS depends on the computational ability of each agent, we consider 
three specific formulations: (1) intersection agents, where each agent controls the allocation of vehicles at a particular intersection; (2) vehicle agents, where each agent controls the allocation of slots for a particular vehicle; (3) atomic agents, where each agent corresponds to a vehicle at a particular intersection. We then compare this to an optimal solution computed using an MILP approach.


In short, we provide the first theoretical analysis of CMMI from the perspective of MAS, making the following novel theoretical and empirical contributions: (1) we prove that computing the optimal solution of CMMI is NP-hard, which shows the essential difficulty of coordination \ijcai{in intersection management}; 
(2) we propose the \ijcai{CMMI game} and a polynomial-time algorithm to find a feasible Nash equilibrium instead of the optimal solution (given some weak assumptions); and (3) we also propose a heuristic that coordinates multiple intersections. The empirical result shows that our algorithm typically achieves a social cost that is only 2.6\% greater than the optimal solution, while being 41\% smaller than the traditional AIM in case of car agents.
\color{black}

The remainder of the paper is structured as follows: Section \ref{sec:model} formalises the problem of CMMI. Section \ref{sec:theory} presents an algorithm that solves CMMI, along with a theoretical analysis of its efficiency. Section \ref{sec:exp} evaluates our algorithm with a real-world data, and Section \ref{sec:conclusion} concludes the paper.

\section{Model}
\label{sec:model}

In this section, we formulate our model of the \ijcai{coordinated management of multiple intersections (CMMI) as a mixed integer program (MIP). Then we transform CMMI into a distributed formulation based on a DCOP game \cite{chapman2011benchmarking}. The notation is summarised in Appendix~\ref{sec:appnote}, and all proofs of theorems are provided in Appendix \ref{sec:appproof}.} 

\subsection{Coordinated Management of Multiple Intersections}
\label{sec:cmmi}

We assume finite time steps $\mc{T}=\{0,\ldots,T\}$, intersections $\mc{I}=\{1,\ldots,I\}$ and cars $\mc{K}=\{1,\ldots,K\}$. The road network is represented by a directed graph $\mc{G} = \langle \mc{I},\mc{E} \rangle$, where edges $\mc{E}\subseteq \mc{I}^2$ denote road segments. The length
of an edge $e \in \mc{E}$ is given by $L_e\in\mathbb{N}_+$. We assume constant and common car speed, each car has a fixed route, and $L_e$ denotes the time to traverse $e$. Also, the edges are not FIFO (First In First Out), but cars can change their order while traversing the edges, if this is beneficial.\footnote[1]{This could be implemented by overtaking or filtering into appropriate lanes. If this is not possible, it is easy to enforce FIFO by introducing additional DCOP constraints.} $\mc{E}_i^{out}=\{(i,j)\in \mc{E}\}$ and $\mc{E}_i^{in}=\{(j,i)\in \mc{E}\}$ are the set of all outgoing and incoming edges of intersection $i$, respectively. Then, the set of edges connected to $i$ is $\mc{E}_i=\mc{E}_i^{out}\cup\mc{E}_i^{in}$.

The goal of CMMI is to minimize total travel time, or equivalently, total waiting time. It involves computing an allocation, i.e. a time slot for each car to pass each intersection on its route. Let $\mc{A}_\ik\subseteq \mc{T}\cup \{\emptyset\}$ denote a set of possible allocations for car $k$ at intersection $i$, where $\emptyset$ means the car is not yet allocated any time slot. 
Since we focus on the externality between intersections, we do not model the detail of the behavior inside the intersections. Then, allocation $a_\ik \in \mc{A}_\ik$ means that car $k$ leaves its incoming edge to intersection $i$ and immediately enters the outgoing edge at $a_\ik$. Each car $k\in\mc{K}$ has a fixed route $r_k=(i^k_1,\ldots,i^k_{\Omega})$ and a departure time $T_k^o$, and we assume that they are common knowledge. Since there is no allocation at the destination, we denote the intersections that car $k$ needs to be allocated as $r_k^{-}=r_k\setminus\{i_{\Omega}^k\}$. \ijcai{Then, the set of variables to be allocated is denoted by $\mc{IK}=\{\ik|k \in \mc{K}, i \in r_k^-\}$ and the solution space is denoted by $\mc{A}=\prod_{\ik \in \mc{IK}}\mc{A}_\ik$.} 
We also denote the previous intersection of $j=i^k_l$ by $\prev(j,k)=i^k_{l-1}$, where $l>1$.
The waiting time of $k$ at $i$ is given by $w_\ik=a_\ik-(a_{\tp{j}{k}}+L_{(j,i)})$, where $a_\ik \in \mc{A}_\ik$, $a_{\tp{j}{k}} \in \mc{A}_{\tp{j}{k}}$ and $j=\prev(i,k)$. If either $a_\ik$ or $a_{\tp{j}{k}}$ is $\emptyset$, $w_\ik=0$. Then, the objective of CMMI is to minimize the total waiting time or total delay $D(a)=\sum_{k\in\mc{K}}D_k(a)$, where $D_k(a)=\sum_{i\in r_k} w_\ik$.

A car 
can collide with other cars when their paths are crossing in an intersection. A path in an intersection $i$ is a pair of an incoming edge $(j,i)$ and an outgoing edge $(i,l)$. We denote the set of all paths passing the intersection by $\mc{P}_i=\{((j,i),(i,l))|(j,i),(i,l) \in \mc{E}_i\}$. A set of paths crossing each other can be denoted by a subset $x_i\subseteq \mc{P}_i$ (Figure \ref{fig:cross}). Then, we denote the set of all crossing path sets in $i$ by $\mc{X}_i\subseteq 2^{\mc{P}_i}$. Let $p_\ik\in\mc{P}_i$ denote the path of $k$ in $i$. Then, cars $k$ and $k'$ collide at intersection $i$ when their paths are crossing ($p_\ik, p_{\tp{i}{k'}}\in x_i$ and $x_i\in\mc{X}_i$) and they are allocated to the same time slot ($a_\ik=a_{\tp{i}{k'}}$). \modified{Also, we denote by $(\perp,(i_1^k,i_2^k))$ a special path that indicates car $k$ departs from $i_1^k$ to $i_2^k$, and this path crosses any path that also moves to $(i_1^k,i_2^k)$ in the end.}

\begin{figure}[t]
\centering
\includegraphics[width=0.75\columnwidth]{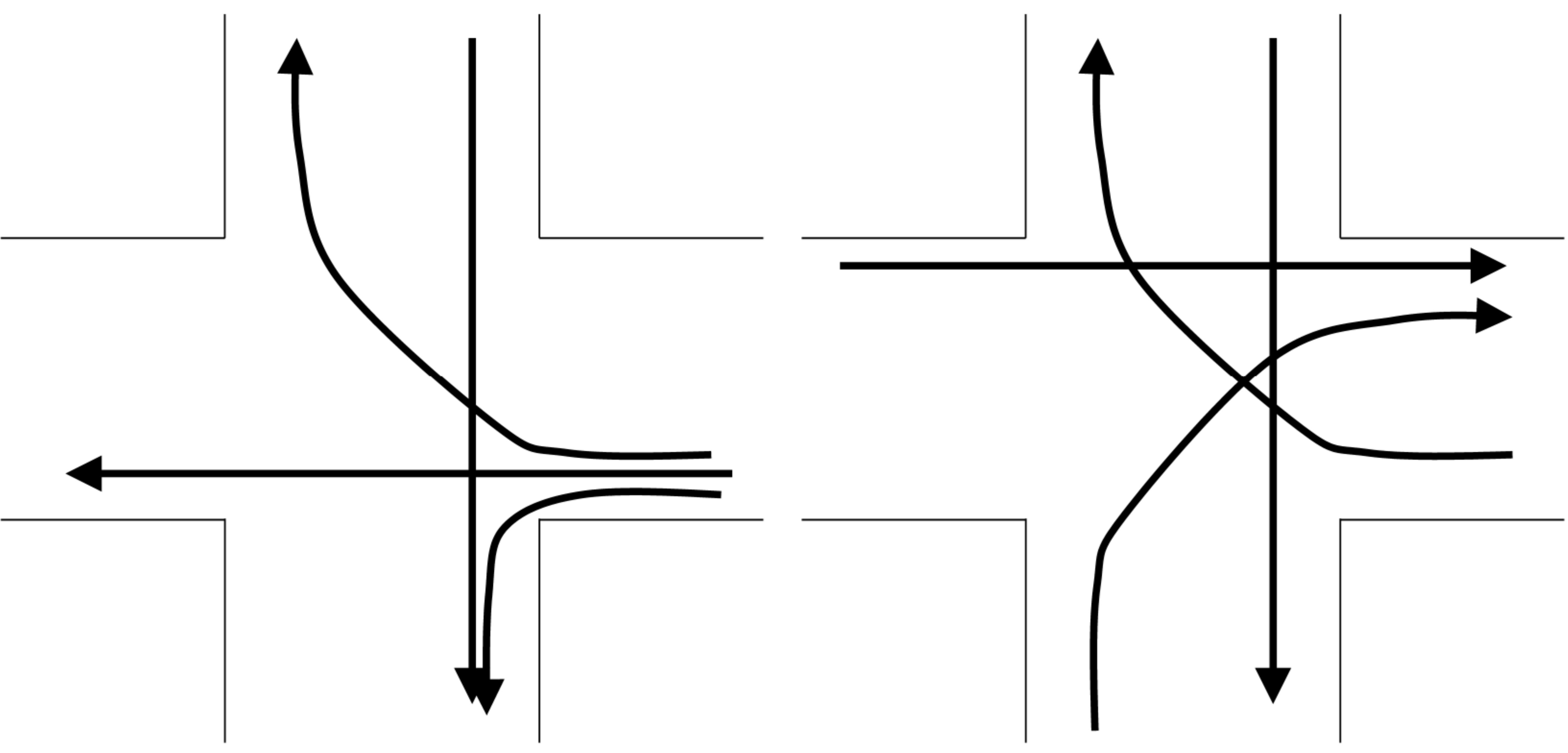}
\caption{Examples of crossing path set $x_i\in\mc{X}_i$ in a typical 4-way intersection. In either case, every pair of the paths in the set is \camera{conflicting} each other.}
\label{fig:cross}
\end{figure}

\color{\newcolor}
Given these definitions, the CMMI problem is formulated as follows:
\begin{equation}
\left.
\begin{array}{l}
\underset{\substack{a \in\mc{A}_{\text{feas}}}}{\min} D(a)
\end{array}
\right.
\label{eq:cmmi}
\end{equation}
where $\mc{A}_{\text{feas}}\subseteq \mc{A}$ is the set of feasible solutions defined by the following constraints:
\begin{align}
\mc{A}_{\text{feas}}=\{a\in\mc{A} &| \nonumber \\
h_{\text{wb}}^{\ik} &: w_\ik \in [0,T_{UB}], ~\forall \ik \in \mc{IK}, \label{eq:hwb} \\
h_{\text{lcap}}^{\tp{t}{e}} &: K_{e}^t\leq K_{UB}, ~\forall t\in\mc{T},\forall e\in\mc{E}, \label{eq:hlcap} \\
h_{\text{dep}}^k &: a_{\tp{i_1^k}{k}}\geq T^o_k,~\forall k\in\mc{K}, \label{eq:hdep} \\
h_{\text{arr}}^k &: a_{\tp{l}{k}}+L_{(l,i^k_{\Omega})}\leq T, l=\prev(i^k_{\Omega},k),\nonumber \\
&\qquad\forall k\in\mc{K}, \label{eq:harr} \\
h_{\text{col}}^i &: X_i=0,~\forall i\in\mc{I}\}, \label{eq:hcol}
\end{align}
which consists of several hard constraints $h$. 
Specifically, constraint (\ref{eq:hwb}) bounds the waiting time where $T_{UB}$ is an upper bound. (\ref{eq:hlcap}) is the constraint of the edge capacity, where $K_{(i,j)}^t=|\{\tp{i}{k} \in \mc{IK}|a_{\tp{i}{k}}\leq t\}|-|\{\tp{j}{k} \in \mc{IK}|a_{\tp{j}{k}}\leq t\}|$ denotes the number of cars on edge $(i,j)$ at time $t$, and $K_{UB}$ is a capacity.
Constraints (\ref{eq:hdep}) and (\ref{eq:harr}) are the departure and arrival time, respectively. Constraint (\ref{eq:hcol}) guarantees no collisions, where $X_i=|\{(\ik,\tp{i}{k'}\in\mc{IK})|k\not=k',\exists x_i\in\mc{X}_i,p_\ik, p_{\tp{i}{k'}}\in x_i, a_{\tp{i}{k}}=a_{\tp{i}{k'}}\}|$ is the number of collisions at $i$. We also provide a MILP formulation of CMMI for linear solvers in Appendix~\ref{sec:appmilp}. 
\color{black}
 Figure~\ref{fig:counterexp} shows an example of CMMI, with $I=7$ and $K=4$. For example, the route of car 2 is $r_2=(7,5,2,3)$ and its path at intersection 5 is $p_{5,2}=((7,5),(5,2))$. 

\color{\newcolor}
CMMI formulates the essential difficulty of the application of intersection management, as we show in the following.
\begin{theorem}
Finding an optimal feasible solution of CMMI that minimizes $D(a)$ is NP-hard, even when all the solutions in $\mc{A}$ satisfy constraints that bound the waiting time, $h_{\text{wb}}, h_{\text{lcap}}$ and $h_{\text{arr}}$.
\label{thm:NPh}
\end{theorem}
This is proved by polynomial-time reduction from an NP-hard job shop scheduling problem to CMMI (in Appendix~\ref{sec:appproof}).
\color{black}

\begin{figure}[t]
\centering
\includegraphics[width=0.5\columnwidth]{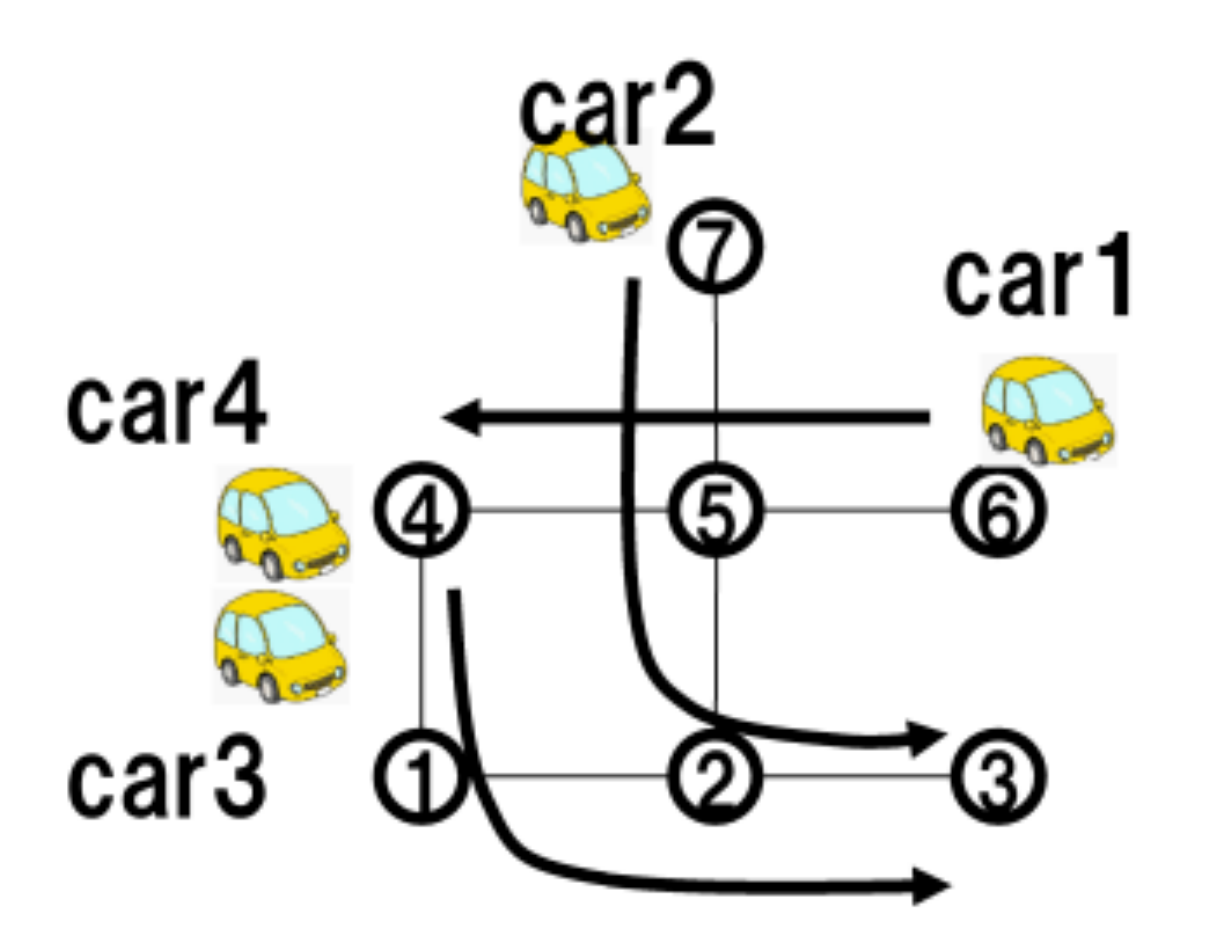}
\caption{Example of CMMI where a cascade of braking can happen due to the externality between multiple intersections: $T^o_1=T^o_2=0, T^o_3=1,T^o_4=2$ and $L_e=5$ for all $e$. \modified{For simplicity, all paths at each intersection are assumed to be crossing (i.e., only one car on each intersection at a time is allowed).
}}
\label{fig:counterexp}
\end{figure}

\subsection{Preliminaries of DCOP and DCOP games}
\label{sec:dcop}

\ijcai{We now formulate a distributed version of CMMI based on DCOP games \cite{chapman2011benchmarking}.} A DCOP game is a game $\langle \mc{N},\mc{A},u \rangle$ corresponding to DCOP $\langle \mc{N},\mc{A},\mc{C} \rangle$, where $\mc{N}$ is a set of agents, $\mc{A}_j$ is the action space of agent $j$, $\mc{A}=\prod_{j \in \mc{N}}\mc{A}_j$ is a joint action space of the agents and $\mc{C}=\{c_1,c_2,\ldots\}$ is a set of constraints. A constraint $c=(\mc{N}_c,f_c)\in\mc{C}$ is defined using $f_c:\prod_{j \in \mc{N}_c}\mc{A}_j\to\mathbb{R}_+$, which is a cost function of the actions by a group of agents, or \emph{neighborhood}, $\mc{N}_c\subseteq\mc{N}$.

Note that the constraints in DCOP and DCOP games are not hard but soft, i.e. using penalties. Hence, additional techniques are required when we apply DCOP / DCOP games to a problem with hard constraints that must be satisfied. 
A DCOP solution is denoted by joint action $a=(a_1,\ldots,a_n)\in\mc{A}$ and we denote a projection of the solution for constraint $c$ as $a^c=(a_j|j\in\mc{N}_c)$. We sometimes denote $f_c(a)=f_c(a^c)$
and also denote $a_{-j}\in\mc{A}_{-j}=\prod_{l\not=j \in \mc{N}}\mc{A}_l$. 
$u=(u_1,\ldots,u_n)$ is a list of utility functions of the agents. Each agent $j\in\mc{N}$ is involved in a set of constraints $\mc{C}_j=\{c|j \in \mc{N}_c\}$, and has a utility function $u_j(a)=-\sum_{c \in \mc{C}_j}f_c(a^c)$. 
The social welfare is defined as $J(a)=\sum_{j \in \mc{N}}u_i(a)$. It is known that every DCOP game is a potential game with potential function $J(a)$ \cite{chapman2011benchmarking}, and has the following desirable property. A best response update (BRU) of an agent is an update of its action assuming the others stay constant, so as to maximize its utility. 
It is known that a finite sequence of unilateral BRUs, where one agent $j$ is selected and executes a BRU for each iteration, always converges to a pure Nash equilibrium $a$ \cite{monderer1996potential}, which satisfies, $\forall j\in\mc{N}, \forall a'_j\in\mc{A}_j,$ and $\forall a_{-j} \in \mc{A}_{-j}$:  
%
%
\begin{equation}
u_j(a_j,a_{-j})\geq u_j(a'_j,a_{-j}).
\label{eq:ne}
\end{equation}

\color{\newcolor}
\subsection{CMMI as a DCOP game}
\label{sec:cmmidcop}
\color{black}

\ijcai{Now we transform the centralized CMMI to the CMMI game, which is a DCOP game.} Since the performance and computational / communication complexity of DCOP games depend on the definition of the agents and the size of their neighborhood, we compare three possible definitions of agents: intersection agents ($\mc{N}=\mc{I}$), car agents ($\mc{N}=\mc{K}$) and atomic agents 
\ijcai{($\mc{N}=\mc{IK}$).}
\ijcai{The action space of atomic agent $\ik$ is the same $\mc{A}_\ik$ defined in Section~\ref{sec:cmmi}. The action space of intersection agent $i$ and car agent $k$ are defined as $\mc{A}_i=\prod_{k|i\in r_k^-}\mc{A}_\ik$ and $\mc{A}_k=\prod_{i\in r_k^-}\mc{A}_\ik$, respectively. 
Note that $\mc{A}_\ik$ is considerably smaller than $\mc{A}_i$ or $\mc{A}_k$.} 

We denote the index function as $\chi$, where $\chi(True)=1, \chi(False)=0$. Also, let $P_{\text{fwd}}$, $P_{\text{col}}$, $P_{\text{none}}$, $P_{2}$ and $P_{1}$ denote the weight constants to prioritise the constraints. These penalties are incurred if the constraints are violated. 
\color{\newcolor}
Then CMMI (\ref{eq:cmmi})-(\ref{eq:hcol}) is converted to constraints of a DCOP game. 
For example, in the case of atomic agents, hard constraint $h_{\text{wb}}^{\ik}$ corresponds to soft constraint $f_{\text{wb}}^{\ik}(a)=\chi(\neg h_{\text{wb}}^{\ik})*P_{1}$, with neighbor $~\mc{N}_{\text{wb}}^{\ik}=\{\tp{\prev(i,k)}{k},\ik\}$ that includes agents involved in $f_{\text{wb}}^{\ik}(a)$. The objective $D(a)$ is also decomposed into a constraint for each agent $f_{\text{delay}}^{\ik}(a)=w_\ik$. As for the other types of agents, the set $\mc{C}$ is same and $f_c$ and $\mc{N}_c$
are the aggregation of composing atomic agents. For example, the soft constraints of all intersections on the route are aggregated, in case of a car agent.
All differences in $u_j$ and $J$ come from the aggregation, which makes the computational complexity large.
Due to the space limit, the full formulation of the three agent models are provided in Appendix~\ref{sec:appconst}.
\color{black}

The minimization of  $D(a)$ coincides with maximizing $J(a)$ in case of car agents as follows:

\begin{lemma}
In the case of car agents, $D(a)=-J(a)$ if $a\in\mc{A}_{\text{feas}}$.
\label{thm:objective}
\end{lemma}

\begin{proof}
If $a$ is feasible, \modified{the utility of car agent $j$ is} $u_j(a)=-\sum_{c\in\mc{C}_j}f_c(a)=f_{\text{delay}}(a)$.
In case of $\mc{N}=\mc{K}$, $f_{\text{delay}}(a)=\sum_{i\in r_k} w_\ik$.
Then, $J(a)=\sum_{k\in\mc{K}}u_j(a)=-\sum_{k\in\mc{K}}\sum_{i\in r_k} w_\ik=-D(a)$.
\end{proof}

In case of intersection agents and atomic agents, $J(a)$ is not \modified{exactly} equal to $-D(a)$ but a weighted sum of $w_\ik$. \modified{This is because $\mc{N}_{\text{delay}}$ contains two agents, $\tp{\prev(i,k)}{k}$ and $\tp{i}{k}$, and their waiting times are double counted when computing the utlities of those agents, in case of atomic agents for example. However, we can still use $-J(a)$ as an approximation of $D(a)$.}

\ijcai{Note that $\mc{A}_{\text{feas}}$ is the set of feasible solutions where the total penalty of all \ijcai{soft} constraints $\mc{C}\setminus\{c_{\text{delay}}\}$ is zero. Also, $\mc{A}_{\text{feas}}$ is the same for all types of agents because if the total penalty is zero for a type of the agents, it remains zero if we change the way of aggregation.}

\ijcai{Since CMMI is NP-\modified{hard}, we do not compute the optimal solution. Instead, a pure Nash equilibrium can be computed using the BRU dynamics. 
However, the feasibility of pure Nash equilibria is not guaranteed by the trivial application of BRU, because they are local optima of $J(a)$. Even worse, it is known that finding a pure Nash equilibrium in a potential game is PLS-complete \cite{fabrikant2004complexity}. We address this issue in Section \ref{sec:theory}.}

\section{Algorithm and Theoretical Analysis}
\label{sec:theory}

\ijcai{To address the difficulty of CMMI discussed above, we propose a distributed solution in the sense that agents do not have to know all the routes of other agents, and they only require the information of available time slots for each intersection at the moment. 
First, we show that our algorithm computes a feasible solution in polynomial time, and the communication complexity is not very intensive either.
We then show a negative result that the solution can be arbitrarily inefficient compared to an optimal solution.} 
Last, we describe how to improve the performance of our algorithm in practice.
\ijcai{Due to the space limit, most of the sub-routines and details of the practical implementation are found in Appendices~\ref{sec:appalgo} and~\ref{sec:appimple}}.

\subsection{Polynomial time feasible algorithm for CMMI}
\label{sec:algo}


To address the challenges above, we propose a novel algorithm \ijcai{based on BRU. Recall that BRU does not guarantee feasible solutions due to the soft constraints.} Since we want to minimize $D(a)$ among feasible solutions, we assume the magnitude of the penalty constants as follows:
\begin{equation}
\left.
\begin{array}{l}
P_{\text{fwd}}\gg P_{\text{col}} \gg P_{\text{none}} \gg P_{2} \gg P_{1}\gg T.
\end{array}
\right.
\label{eq:penalty}
\end{equation}

\begin{algorithm}[t]                 
	\caption{Best response update with downstream reset}
	\label{alg:brudr}
	\begin{algorithmic}[1]
	\Procedure{$a=$BRUDR}{$a,bestU$}
		\While{$a$ violates (\ref{eq:ne})} \label{brudr:ne}
		\State{$\hat{\Delta}=0,\hat{a}=\emptyset$}
    	\For{$j \in \mc{N}$} \label{brudr:itr1}
    	\For{$a'_j \in \mc{A}_j$} \label{brudr:space}
    	\State{$a''=\textproc{DownReset}(j,a,(a'_j,a_{-j}))$} \label{brudr:reset}
    	\State{$\Delta u_j=u_j(a'')-u_j(a)$} \label{brudr:improve}
    	\If{$\Delta u_j>0$}
    	\If{not $bestU$}
		\State{$a=a''$} \label{brudr:upd1}
		\State{break}
    	\ElsIf{$\Delta u_j>\hat{\Delta}$}
	    \State{$\hat{\Delta}=\Delta u_j,\hat{a}=a''$} \label{brudr:bestupd}
		\EndIf
		\EndIf
		\EndFor
		\EndFor
    	\If{$bestU$ and $\hat{a}\not= None$}
    	\State{$a=\hat{a}$} \label{brudr:upd2}
    	\EndIf
		\EndWhile
		\State{Return $a$} \label{brudr:ret}
	\EndProcedure
	\end{algorithmic}
\end{algorithm}

Since BRU tries to find a local maximum of $J(a)$,
(\ref{eq:penalty}) gives priority to obtaining a feasible solution before minimizing $D(a)$. 
\ijcai{In effect, we apply the ``big M'' method to DCOPs, using it to convert hard constraints to large penalty functions, or soft constraints, in the DCOP objective function.}
Our algorithm, BRU with downstream reset (\textproc{BRUDR}) is shown in Algorithm \ref{alg:brudr}.
Briefly, every time the algorithm allocates a time slot to a car at an intersection, it clears all allocations of the car at  downstream intersections to prevent the violation of $h_{\text{wb}}$.
We initialize the joint action as being unallocated, $a=\emptyset$, which means $a_\ik=\emptyset, \forall \ik \in \mc{I}\times\mc{K}$. 

We consider two variants of the algorithm, one which always updates the agent with the largest improvement in utility (indicated by $bestU=True$) and one which updates any agent ($bestU=False$). The algorithm terminates when the joint action converges to a pure Nash equilibrium (line \ref{brudr:ne}). Otherwise, it computes the improvement in utility of the update for each action of each agent (line \ref{brudr:itr1}-\ref{brudr:improve}). \modified{Without loss of generality, the agent loop is ordered by agent ID (line \ref{brudr:itr1}). Also, the action loop is in ascending order of the value $a_j'$ (line \ref{brudr:space}).}
The types of agents make a difference here because the different size of neighborhood causes a different aggregation of the costs. In particular, only car agents have a neighborhood over all intersections on their own routes. After that, different from the normal BRU, the algorithm resets the allocation of the downstream intersections of the current agent $j$ to $\emptyset$ using the function \textproc{DownReset} (line \ref{brudr:reset}). \modified{Without this process, a car can violate $h_{\text{wb}}$ by creating an infeasible schedule between different intersections, especially in the case of atomic agents with a non-empty initial allocation.}
If $bestU=False$, the joint action is updated when there is an improvement (line \ref{brudr:upd1}). If $bestU=True$, the algorithm finds an update with the largest improvement (line \ref{brudr:bestupd}), and applies the update (line \ref{brudr:upd2}). \modified{Finally, the algorithm returns a Nash equilibrium (line \ref{brudr:ret}).} 
With the resetting, BRUDR works as a greedy algorithm that allocates time slots from the agents in the upstream, until all agents are allocated. BRUDR guarantees the feasibility of the solution as follows.

\begin{lemma}
In cases of the car agents and atomic agents, BRUDR converges to a feasible pure Nash equilibrium, with each agent updating the action at most once, if all the solutions in $\mc{A}$ satisfy $h_{\text{wb}}, h_{\text{lcap}}$ and $h_{\text{arr}}$.
\label{thm:feas}
\end{lemma}

Briefly, this is because agents can always find the earliest allocation that satisfies the \ijcai{soft} constraints, because agents can delay their trips as much as they need \modified{assuming that the 3 bounding constraints, $h_{\text{wb}}, h_{\text{lcap}}$ and $h_{\text{arr}}$ are always satisfied. In practice, this assumption holds when the traffic is not so dense.}  In case of intersection agents, however, the resetting mechanism can unallocate the other cars on their downstream route, making the feasible updates impossible.


Note that Lemma \ref{thm:feas} means that the efficiency of BRUDR depends only on the allocation order of the agents, because all agents choose their allocation only once. As for the computational complexity of BRUDR, it suffers from an exponential size of the search space for the cases of intersection agents and car agents. To address this, we restrict the action space. For each intersection on the route, the car checks the feasibility of each time slot, and allocates the earliest one that satisfies all the \ijcai{soft} constraints.
Now we can solve CMMI in polynomial time as follows.

\begin{theorem}
With the restricted action space, BRUDR converges to a feasible pure Nash equilibrium in polynomial time if all solutions in $\mc{A}$ satisfy $h_{\text{wb}}, h_{\text{lcap}}$ and $h_{\text{arr}}$.
\label{thm:poly}
\end{theorem}

\begin{proof}
This follows from Lemma \ref{thm:feas}. Since all agents update at most once, the number of \modified{the iterations} in BRUDR is less than \modified{$T*I*K$}. The size of the action space is polynomial in \modified{$I*T$} multiplied by the complexity of calculating $f_c$. As in the definition of constraints, the complexity of calculating $f_c, \forall c\in\mc{C}$ is also polynomial.
\end{proof}

\modified{Theorem \ref{thm:poly} shows that our algorithm relaxes the problem of finding a Nash equilibrium to an easier complexity class than PLS-complete, by assuming the three constraints $h_{\text{wb}}, h_{\text{lcap}}$ and $h_{\text{arr}}$ are satisfied.} In practice, the algorithm can be accelerated even further by replacing the best response  update of a single player with the stochastic response updates of multiple players. For each iteration, all players \modified{simultaneously} try to choose their best response to the joint actions in the former iteration, but succeed only with some probability $p$. This is also known as the distributed stochastic algorithm (DSA) and it almost surely converges to a Nash equilibrium in potential games and DCOP games \cite{chapman2011benchmarking}.

\ijcai{Note that it is easy to apply our approach to dynamic settings, because of the greedy nature of the distributed algorithms. The agents can simply run \textproc{BRUDR} again when a new vehicle appears or some vehicles change their routes. Since it is the same as if the vehicles existed with the new routes from the beginning (but which had not yet been updated), the same proof for Theorem \ref{thm:poly} would hold as well. Also, the communication complexity is not very intensive in our approach. See summary in Appendix \ref{sec:appcomm}.}

\ijcai{Despite of the scalability above, the efficiency of the solution can be arbitrarily worse than the optimal one. We discuss the detail in Appendix \ref{sec:apppoa}, but Figure \ref{fig:counterexp} provides an illustrative example. In the optimal solution, $D(a_{so})=1$ where all cars do not have to brake except for car 1. However, \textproc{BRUDR} can output a Nash equilibrium where all cars except for car 1 brake, resulting in $D(a_{ne})=3$. Since we can add more cars after car 4, $D(a_{ne})$ can be arbitrarily worse than $D(a_{so})$.}

\subsection{Personal Best Actions}
\label{sec:pba}

Although the inefficiency of CMMI cannot be bounded, in practice, it is still possible to coordinate the intersections and improve efficiency. We propose a \modified{technique} that can be applied to a wide range of multiagent coordination problems. We can improve the efficiency by just replacing the initial input of BRUDR with the Personal Best Actions (PBA).

\begin{equation}
\left.
\begin{array}{l}
a_j=\underset{\substack{a'_j\in \mc{A}_j}}{\argmax} u_j(a'_j,\emptyset_{-j}), \forall j\in\mc{N},
\end{array}
\right.
\label{eq:pba}
\end{equation}

\noindent where $\emptyset_{-j}$ means $a_l=\emptyset, \forall l \in -j$. In CMMI, this means that agents choose the earliest time slots without worrying about collisions. 

PBA can improve the efficiency compared to the default BRUDR (with initial solution $a=\emptyset$) on average across all instances of CMMI. We show this empirically in the later experimental section, but here we briefly explain why PBA performs better than the default BRUDR. Let $\mc{F}$ be the finite set of all the instances of CMMI. Then we split  $\mc{F}$ into the two disjoint subsets as $\mc{F}=\mc{F}_1\sqcup\mc{F}_2$. 
The set $\mc{F}_1$ consists of special cases where a single update of an agent from PBA can reach the best Nash equilibrium.
All of the other instances fall into $\mc{F}_2$.
In case of $\mc{F}_1$, BRUDR starting from PBA always outputs the best Nash equilibrium.
Note that the same thing never happens with the default BRUDR, because it cannot terminate in one step from $a=\emptyset$ when $N\geq 2$. Then, BRUDR with PBA always outperforms the default BRUDR for $\mc{F}_1$. For $\mc{F}_2$, both BRUDR with PBA and the default BRUDR cannot guarantee the optimality, and show the same performance on average. Then, BRUDR with PBA always outperforms the default BRUDR across all instances $\mc{F}$ on average. Since the idea of PBA is simple, we can apply it to other resource allocation problems.

\section{Experimental Results}
\label{sec:exp}

Given that we do not have any efficiency bounds, we evaluate BRUDR using simulations \ijcai{with the data set of vehicle trajectories in Athens that were collected as part of the pNEUMA field experiment \cite{barmpounakis2020new}
\footnote[2]{We thank EPFL Urban Transport Systems Laboratory for the data source: pNEUMA – open-traffic.epfl.ch}.
We analyse a zone of Athens including 24 intersections and extract the vehicle trajectories that pass those intersections to compute the routes of the vehicles. Based on this, we create a traffic demand model as a Poisson distribution for each route, with the mean parameter estimated from the route data. We generate cars for each time slot randomly according to the Poisson distribution, until the total number of cars reaches the parameter $K$.} We simulate different levels of car density, by changing the parameter $K$.
We denote the default BRUDR and BRUDR with PBA as $default$ and $pba$, respectively.
Since our interest is in the problem instances where cars conflict with each other, we exclude cases without conflicts, i.e., where $D(a_{so})=0$. We also exclude the easy cases where both $pba$ and $default$ reach $a_{so}$.

In what follows we first compare the performance of BRUDR and the traditional first-come-first-served AIM \modified{\cite{dresner2008multiagent}} with different settings. We then evaluate the scalability of the algorithm to larger settings. All the error bars in the figures show 95\% confidence intervals.

\modified{For reproducibility of the experiments, the details of the setting are in Appendix \ref{sec:appexp}. The code is also available in the supplementary material.}

\subsection{Small Settings}
\label{sec:expsmall}

We compare the performance of BRUDR with an optimal solution by changing the definitions of agents ($\mc{K}$, $\mc{I}$ and $\mc{I}\times\mc{K}$), with different parameters $K\in\{20,30\}$, and different initial allocations (PBA and $a=\emptyset$). \ijcai{Since finding an optimal solution is NP-hard, we only focus on the small size problems.} Note that the case of the atomic agents with the initial actions $a=\emptyset$ corresponds to the traditional AIM that uses a FCFS allocation policy \cite{dresner2008multiagent}.
We run 300 simulations for each setting.

Figure \ref{fig:smalleff} shows the  ratio of total delay compared to the optimum one. Since solutions can be infeasible in case of intersection agents, it includes only the result of feasible solutions. \modified{As discussed in Section \ref{sec:pba}}, the delay of $pba$ becomes smaller than the one of $default$ and close to the optimum, in case of the car agents.
Note that there is no clear improvement in case of intersection agents. This is because the improvement requires updates at multiple intersections and cannot be done by a single update at an intersection, so $\mc{F}_1$ is almost empty. Hence, this result shows that the case of car agents enables coordination of multiple intersections and exceeds the performance of the single intersection management (about \ijcai{41}\% \camera{points}
reduction on average\modified{, while it is \ijcai{2.6}\% \camera{points} greater than the delay of optimal solutions \camera{when $K=20$}}). In case of car agents, $pba$ achieves the social optimum even in the pathological case of Example \ref{thm:example}, \modified{maintaining the platoon} (car 3 and 4).
However, the ratio becomes smaller when $K$ increases, in all settings. Since larger $K$ causes more conflicts among cars, this indicates that even the optimum traffic deteriorates by congestion close to Nash equilibria, and there is little room remained to improve.
However, cars can have negative externalities that cause cascades of braking as in Figure \ref{fig:counterexp}, even in large problems. We analyze this aspect further in Appendix \ref{sec:appextern}.

\ijcai{Since the objective of CMMI is minimization of total delay $D(a)$, it can cause unfair allocations. We show an interesting finding that the fairness depends on the agent models.} Figure \ref{fig:smallgini} shows the Gini index, which is a degree of unfairness among cars. 
\camera{PBA with the intersection agents is fairer than PBA with the other agents with 95\% significance,} because all cars that pass an intersection join a neighborhood and share the common utility. 
This indicates there may be a tradeoff between efficiency and fairness, \ijcai{since intersection agents are more fair but are worse in social welfare, while car agents have the highest social welfare but lower fairness. This tradeoff can also be} seen in Example \ref{thm:example}, and as is often observed MAS in allocation problems in general.

\begin{figure}[t]
\centering
\includegraphics[width=0.85\columnwidth]{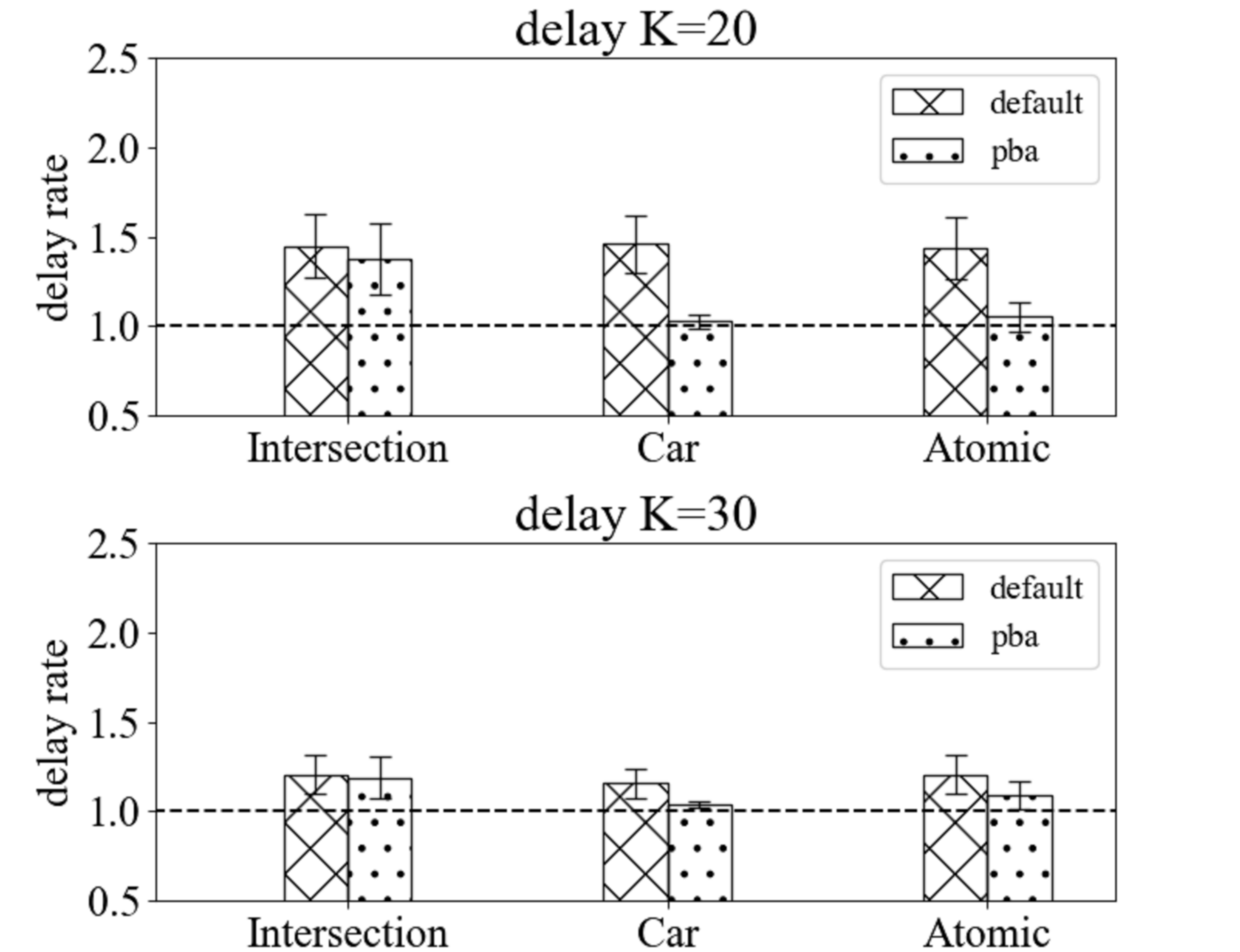}
\caption{The ratio of total delay compared to the optimum one $D(a)/D(a_{so})$. Comparisons between $default$ and $pba$ with different types of agents and density of cars $K/I$.}
\label{fig:smalleff}
\end{figure}

\begin{figure}[t]
\centering
\includegraphics[width=0.85\columnwidth]{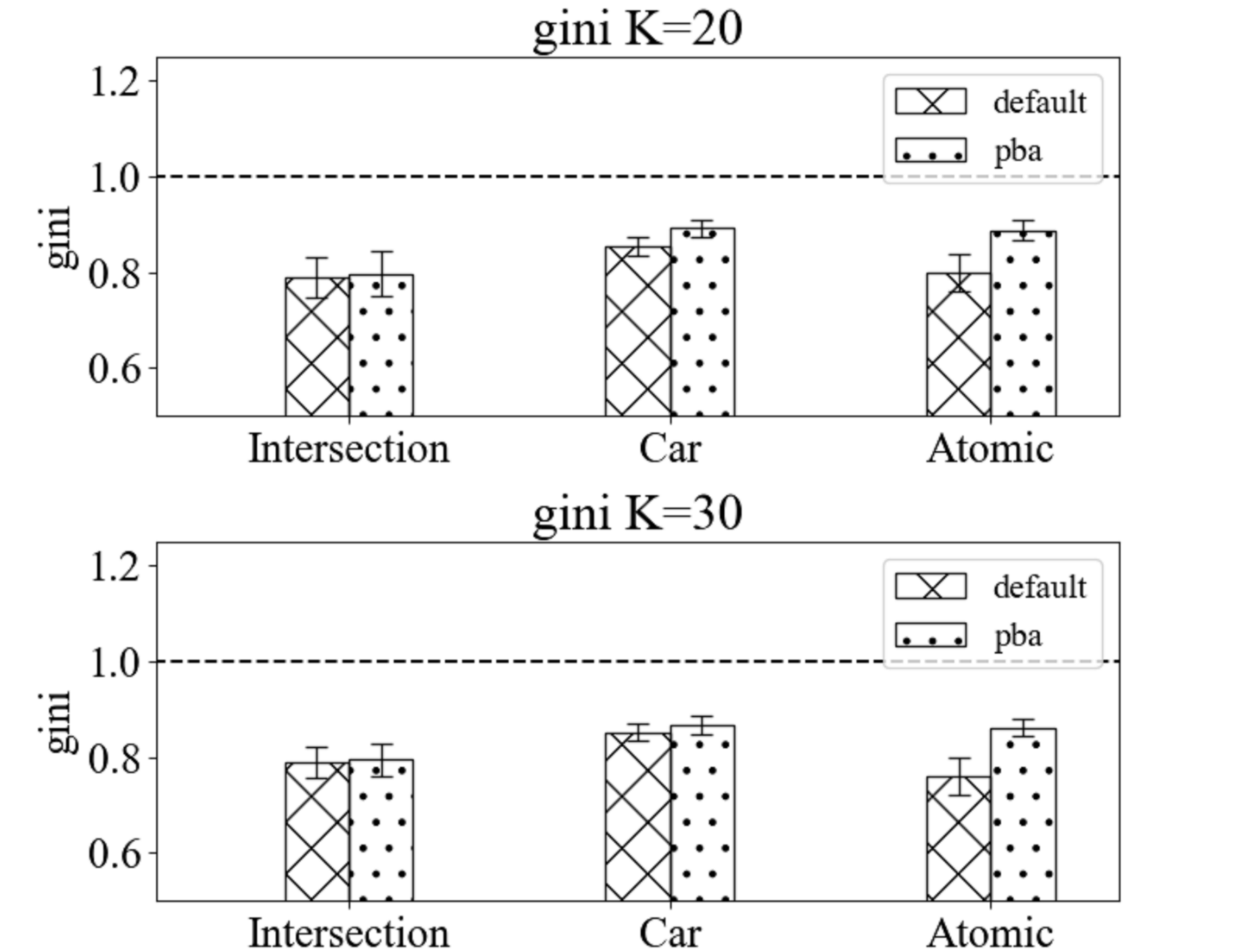}
\caption{Gini index $GI(a)=\sum_{k\in\mc{K}}\sum_{k'\in\mc{K}}|D_k(a)-D_{k'}(a)|/2K|D(a)|$. Comparisons between $default$ and $pba$ with different types of agents and density of cars $K/I$.}
\label{fig:smallgini}
\end{figure}

\subsection{Scalability}
\label{sec:scale}

Next, we evaluate the scalability of BRUDR in case of car agents. We change the number of cars $K\in\{10,20,30,40,50,70,90,110,130\}$, and run 100 simulations for each case. 
Since most of the runtime is spent on computing $c_{\text{lcap}}$, we compare the cases with and without the constraint. Figure \ref{fig:large} shows the runtime of BRUDR.
Though it shows that our algorithm is polynomial time in $K$, the runtime can exceed $10^4$ seconds if the computation of $c_{\text{lcap}}$ is included. However, if we can assume the traffic is not so dense and can omit the computation of $c_{\text{lcap}}$, our algorithm is scalable to more than 100 cars.
Also, PBA contributes to reducing the runtime because all agents are already allocated in the initial solution. 
\ijcai{Note that the current implementation uses a single CPU just for theoretical verification.
However, our polynomial-time algorithm can be easily executed in a distributed fashion and therefore scale, technically.}

\begin{figure}[t]
\centering
\includegraphics[width=0.85\columnwidth]{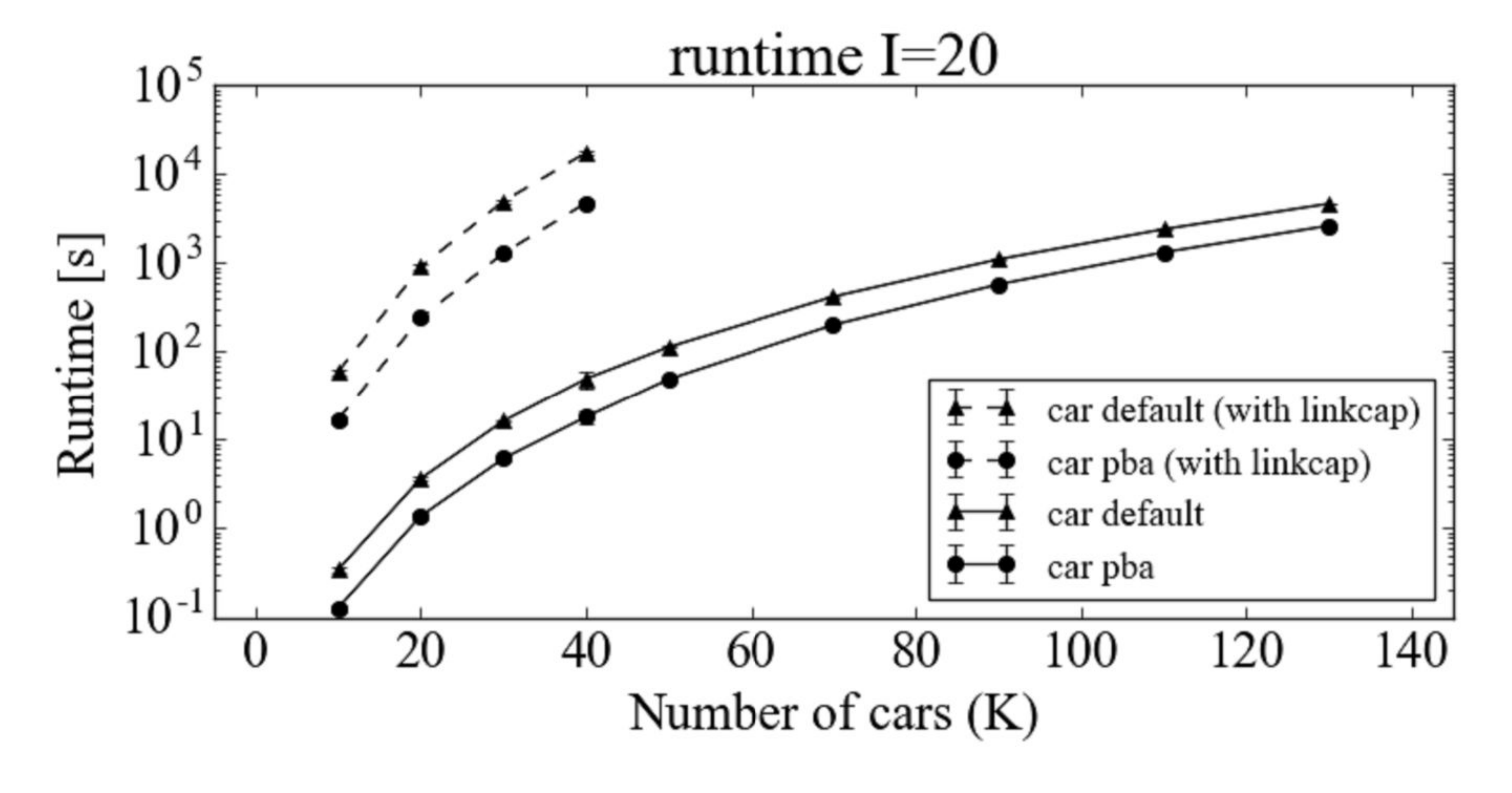}
\caption{Runtime with different number of cars $K$. Comparisons between $default$ and $pba$, with and without $c_{\text{lcap}}$ computation.}
\label{fig:large}
\end{figure}

\section{Conclusions}
\label{sec:conclusion}

We model the coordination problem of multiple intersections with the theory of DCOP. The goal is the theoretical analysis of the problem, and to find an efficient feasible allocation in polynomial time. To this end, we present a polynomial-time algorithm that guarantees to find a feasible solution and we present a novel modification to improve its efficiency in practice. In addition, through numerical simulation using real-world data, we show that our algorithm outperforms traditional single intersection management approaches. \modified{Since our problem is a special case of the job shop scheduling problem, our work may also be applicable to a wide variety of scheduling problems.}
Future directions include an algorithm that bounds the inefficiency with communications among agents, as well as extending the problem \modified{to include vehicle routing rather than assuming fixed routes.}

 
\newpage

\section{Ethics statement}

\ijcai{Our work has positive ethical impacts on society. This is because our algorithm can benefit society and the environment by reducing total travel time and pollution. A possible negative impact is that people might feel unfair because some vehicles have to wait more than the others. We could address this by looking into envy-free algorithm in future work.}
\bibliographystyle{named}
\bibliography{ref}

\if\arxiv1

%


\newpage
\appendix

\color{\newcolor}

\section{Summary of notation}
\label{sec:appnote}

\begin{tabular}{|l|l|}
\hline
Time & $t\in\mc{T}$ \\ 
Intersection & $i\in\mc{I}$ \\ 
Car & $k\in\mc{K}$ \\ 
Edge (road segment) & $e\in\mc{E}$ \\
Length of edge & $L_e$\\
Allocation & $a_\ik \in \mc{A}_\ik \subseteq \mc{T}\cup\{\emptyset\}$\\
Route of car & $r_k = (i_1^k,\ldots, i_{\Omega}^k)$\\
 & $r_k^- = r_k\setminus \{i_{\Omega}^k\}$\\
Departure time of car & $T_k^o$\\
Previous intersection of $i^k_l$ & $\prev(i,k)=i^k_{l-1}$\\
Waiting time & $w_\ik$\\
Total waiting time & $D(a)=\sum_{k\in\mc{K}}\sum_{i\in r_k} w_\ik$\\
All paths passing $i$ & $\mc{P}_i$ \\ 
Path set crossing each other & $x_i\in \mc{X}_i$\\
Constraint of DCOP game & $c=(\mc{N}_c,f_c) \in \mc{C}$\\
Utility function & $u_i$\\
\hline
\end{tabular}

\section{MILP formulation of CMMI}
\label{sec:appmilp}

Variables of MILP $y$ consists of $y_{t,k,e}\in\{0,1\}, ~\forall t\in\mc{T},\forall k\in\mc{K},\forall e\in\mc{E}_k$, 
where $\mc{E}_k=\{(i,j)|i \in r_k^-\}$.
The value of $y_{t,k,e}=1$ 
if and only if $a_\ik=t$. 
The wait time of car $k$ to enter $e=(i,j)$ is: \begin{align*}
c(y,k,e)=&\underset{\substack{t_{out} \in \mc{T}}}{\sum}t_{out}*y_{t_{out},k,e} \\
&-\underset{\substack{t_{in} \in \mc{T}}}{\sum}t_{in}*y_{t_{in},k,(\prev(k,i),i)}-L_{(\prev(k,i),i)}.
\end{align*}
Given this, the MILP formulation is as follows:
\begin{align}
&\underset{\substack{y}}{\min} \underset{\substack{k \in \mc{K}}}{\sum}\underset{\substack{e \in \mc{E}_k}}{\sum}c(y,k,e)  \label{eq:mj} \\
& \text{subject to:} \nonumber\\
&c(y,k,e)\geq 0,~\forall k \in \mc{K}, \forall e\in \mc{E}_k, \label{eq:mwlb} \\
&c(y,k,e)\leq T_{UB},~\forall k \in \mc{K}, \forall e\in \mc{E}_k, \label{eq:mwub} \\
&\underset{\substack{k \in \mc{K}_e}}{\sum} y_{t,k,e}\leq K_{UB}~\forall t \in \mc{T},~\forall e\in \mc{E}, \label{eq:mlcap} \\
&y_{T_k^o,k,e(i_1^k,i_2^k)}=1,~\forall k\in\mc{K}, \label{eq:mdep} \\
&y_{t,k,(i_{\Omega-1}^k,i_{\Omega}^k)}=0,~\forall k\in\mc{K},\forall t>T-L_{(i_{\Omega-1}^k,i_{\Omega}^k)}, \label{eq:marr} \\
&\underset{\substack{(e_{in},e) \in x_i}}{\sum}\underset{\substack{k \in \mc{K}|e_{in},e\in \mc{E}_k}}{\sum}y_{t,k,e}\leq 1~\forall t \in \mc{T},~\forall i\in \mc{I},~\forall x_i \in\mc{X}_i, \label{eq:mcol}\\
&\underset{\substack{t \in \mc{T}}}{\sum}x_{t,k,e}= 1,~\forall k \in \mc{K},~\forall e\in \mc{E}_k,\label{eq:mtex}\\
&\underset{\substack{k \in \mc{K}}}{\sum}x_{t,k,e}\leq 1,~\forall t \in \mc{T},~\forall e\in \mc{E},\label{eq:mvex}\\
&\underset{\substack{e \in \mc{E}}}{\sum}x_{t,k,e}\leq 1,~\forall t \in \mc{T},~\forall k\in \mc{K}.\label{eq:meex}
\end{align}
where $\mc{K}_e=\{k\in\mc{K}|e\in \mc{E}_k\}$. Constraints (\ref{eq:mwlb})-(\ref{eq:mcol}) correspond to (\ref{eq:hwb})-(\ref{eq:hcol}) in Section~\ref{sec:model}. Constraints (\ref{eq:mtex})-(\ref{eq:meex}) are newly added to guarantee that only one vehicle is allocated at a time at an edge.

\section{Constraints for three agent models}
\label{sec:appconst}

We define a set of cars that pass edge $(i,j)$ as: \[\mc{K}_{(i,j)}=\{k\in\mc{K}|i\in r_k^-,j\in r_k^-\}.\] 
Also, we define the following hard constraints that apply to to all three models.
First, car $k$ can be allocated at $i$ only when it is allocated at the upstream intersection: $\prev(i,k)$
\begin{align}
&h_{fwd}^{\ik}:~a_\ik  =\emptyset \lor a_{\tp{\prev(i,k)}{k})}\not=\emptyset,~\forall \ik\in\mc{IK} \label{eq:hfwd}
\end{align}
Second, car $k$ is allocated at $i$:
\begin{align}
&h_{alloc}^{\ik}(a):~a_{\ik}\not =\emptyset,~\forall \ik\in\mc{IK} \label{eq:halloc}
\end{align}

\subsection{Constraints for atomic agents}

For the case of atomic agents, the set of constraints $\mc{C}$ is formulated as follows:
\begin{align}
c_{delay}^{\ik}:&~\mc{N}_{delay}^{\ik}=\{\tp{\prev(i,k)}{k},\ik\}, \nonumber \\
&f_{delay}^{\ik}(a)=w_\ik,~\forall \ik \in \mc{IK}, \label{eq:cdelay} \\
c_{wb}^{\ik}:&~\mc{N}_{wb}^{\ik}=\{\tp{\prev(i,k)}{k},\ik\}, \nonumber \\
&f_{wb}^{\ik}(a)=\chi(\neg h_{wb}^{\ik})*P_{1},~\forall \ik \in \mc{IK}, \label{eq:cwb} \\
c_{lcap}^{\tp{t}{e}}:&~\mc{N}_{lcap}^{\tp{t}{e}}=\{\tp{i}{k},\tp{j}{k}|k\in\mc{N},i=\prev(j,k)\}, \nonumber \\
&f_{lcap}^{\tp{t}{e}}(a)=\chi(\neg h_{lcap}^{\tp{t}{e}})*P_{1},~\forall t\in\mc{T},\forall e\in\mc{E},  \label{eq:clcap} \\
c_{dep}^k:&~\mc{N}_{dep}^k=\{\tp{i_1^k}{k}\},\nonumber \\
&f_{dep}^k(a)=\chi(\neg h_{dep}^k)*P_{2},~\forall k\in\mc{K}, \label{eq:cdep} \\
c_{arr}^k:&~\mc{N}_{arr}^k=\{\tp{\prev(i^k_{\Omega},k)}{k}\},\nonumber \\
&f_{arr}^k(a)=\chi(\neg h_{arr}^k)*P_{1},~\forall k\in\mc{K}, \label{eq:carr} \\
c_{col}^i:&~\mc{N}_{col}^i=\{\ik \in \mc{IK}\},\nonumber \\
&f_{col}^i(a)=\chi(\neg h_{col}^i)*P_{col},~\forall i\in\mc{I}, \label{eq:ccol} \\
c_{fwd}^{\ik}:&~\mc{N}_{fwd}^{\ik}=\{\tp{\prev(i,k)}{k},\tp{i}{k}\},\nonumber \\
&f_{fwd}^{\ik}(a)=\chi(\neg h_{fwd}^{\ik})*P_{fwd} \label{eq:cfwd}\\
c_{alloc}^{\ik}:&~\mc{N}_{alloc}^{\ik}=\{\ik\},\nonumber \\
&f_{alloc}^{\ik}(a)=\chi(\neg h_{alloc}^{\ik})*P_{none}. \label{eq:calloc}
\end{align}

The correspondence between the hard constraints $h$ (\ref{eq:hwb})-(\ref{eq:hcol}) and the soft constraints $f_c$ (\ref{eq:cwb})-(\ref{eq:ccol}) is straightforward. Also, (\ref{eq:hfwd}) and (\ref{eq:halloc}) correspond to (\ref{eq:cfwd}) and (\ref{eq:calloc}). The neighborhood $\mc{N}_c$ of each constraint includes agents that are involved in the soft constraint. (\ref{eq:cdelay}) corresponds to the objective $D(a)$. 

\subsection{Constraints for car agents}

\begin{align}
c_{delay}^{k}:&~\mc{N}_{delay}^{k}=\{k\}, \nonumber \\
&f_{delay}^{k}(a)=\underset{\substack{i \in r_k^-}}{\sum}w_\ik,~\forall k \in \mc{K}, \label{eq:ccdelay} \\
c_{wb}^{k}:&~\mc{N}_{wb}^{k}=\{k\}, \nonumber \\
&f_{wb}^{k}(a)=(1-\underset{\substack{i \in r_k^-}}{\prod}\chi( h_{wb}^{\ik}))P_{1},~\forall k \in \mc{K}, \label{eq:ccwb} \\
c_{lcap}^{\tp{t}{e}}:&~\mc{N}_{lcap}^{\tp{t}{e}}=\mc{K}_{(i,j)}, \nonumber \\
&f_{lcap}^{\tp{t}{e}}(a),~\forall t\in\mc{T},\forall e=(i,j)\in\mc{E},  \label{eq:cclcap} \\
c_{dep}^k:&~\mc{N}_{dep}^k=\{k\},\nonumber \\
&f_{dep}^k(a),~\forall k\in\mc{K}, \label{eq:ccdep} \\
c_{arr}^k:&~\mc{N}_{arr}^k=\{k\},\nonumber \\
&f_{arr}^k(a),~\forall k\in\mc{K}, \label{eq:ccarr} \\
c_{col}^i:&~\mc{N}_{col}^i=\{k \in \mc{K}| i\in r_{k}^-\},\nonumber \\
&f_{col}^i(a),~\forall i\in\mc{I}, \label{eq:cccol} \\
c_{fwd}^{k}:&~\mc{N}_{fwd}^{k}=\{k\},\nonumber \\
&f_{fwd}^{k}(a)=(1-\underset{\substack{i \in r_k^-}}{\prod}\chi(h_{fwd}^{\ik}))P_{fwd},~\forall k\in\mc{K} \label{eq:ccfwd}\\
c_{alloc}^{k}:&~\mc{N}_{alloc}^{k}=\{k\},\nonumber \\
&f_{alloc}^{k}(a)=(1-\underset{\substack{i \in r_k^-}}{\prod}\chi(h_{alloc}^{\ik}))P_{none},~\forall k\in\mc{K}. \label{eq:ccalloc}
\end{align}

\subsection{Constraints for intersection agents}

\begin{align}
c_{delay}^{(i,j)}:&~\mc{N}_{delay}^{(i,j)}=\{i,j\}, \nonumber \\
&f_{delay}^{(i,j)}(a)=\underset{\substack{k \in \mc{K}_{(i,j)}}}{\sum}w_\ik,~\forall (i,j) \in \mc{E}, \label{eq:cidelay} \\
c_{wb}^{(i,j)}:&~\mc{N}_{wb}^{(i,j)}=\{i,j\}, \nonumber \\
&f_{wb}^{(i,j)}(a)=(1-\underset{\substack{k \in \mc{K}_{(i,j)}}}{\prod}\chi( h_{wb}^{\ik}))P_{1},~\forall (i,j) \in \mc{E}, \label{eq:ciwb} \\
c_{lcap}^{\tp{t}{e}}:&~\mc{N}_{lcap}^{\tp{t}{e}}=\{i,j\}, \nonumber \\
&f_{lcap}^{\tp{t}{e}}(a),~\forall t\in\mc{T},\forall e=(i,j)\in\mc{E},  \label{eq:cilcap} \\
c_{dep}^k:&~\mc{N}_{dep}^k=\{i_1^k\},\nonumber \\
&f_{dep}^k(a),~\forall k\in\mc{K}, \label{eq:cidep} \\
c_{arr}^k:&~\mc{N}_{arr}^k=\{\prev(i^k_{\Omega},k)\},\nonumber \\
&f_{arr}^k(a),~\forall k\in\mc{K}, \label{eq:ciarr} \\
c_{col}^i:&~\mc{N}_{col}^i=\{i\},\nonumber \\
&f_{col}^i(a),~\forall i\in\mc{I}, \label{eq:cicol} 
\end{align}
\begin{align}
c_{fwd}^{i}:&~\mc{N}_{fwd}^{k}=\{i,j\},\nonumber \\
&f_{fwd}^{e}(a)=(1-\underset{\substack{k \in \mc{K}_{(i,j)}}}{\prod}\chi(h_{fwd}^{\ik}))P_{fwd},
\nonumber \\ & \hspace{110pt} \forall e=(i,j)\in\mc{E} \label{eq:ccfwd}\\
c_{alloc}^{i}:&~\mc{N}_{alloc}^{i}=r_{k}^-,\nonumber \\
&f_{alloc}^{i}(a)=(1-\underset{\substack{k|i\in r_{k}^-}}{\prod}\chi(h_{alloc}^{\ik}))P_{none},\forall i\in\mc{I}. \label{eq:ccalloc}
\end{align}
\color{black}

\section{Practical Implementation}
\label{sec:appimple}

For example, the infrastructure and the system can work as follows. When a driver sets the destination into a car navigation system, the car communicates with the intersections on its route and asks for the availability of the slots at that moment. Then the car checks which slots can satisfy the constraints of CMMI, and reserves feasible slots at relevant intersections. This computation can be done in a distributed manner by each car. The intersections work as a database that store the usage of the slots, and mediate the communication between cars by sharing the information of available slots. 
The advantage of this distributed design is threefold. First, 
it is computationally scalable in dynamic setting. When a new vehicle appears in the environment or some vehicles change their routes, only the agents influenced by those changes has to update their own trips, while a centralized approach has to compute the entire allocation. Second, it is preferable from a privacy perspective, because we do not want cars to have much information about the routes of other cars. Third, there is the real-world robustness of having optimisation nodes scattered on edge computing devices across a physical network,
rather than a centralised optimisation node that can be a single point of failure. Furthermore, our approach can be applied to the mixed traffic of manual and autonomous vehicles, because it also works for connected (but not necessarily autonomous) vehicles. In that case, the agent can be on e.g. the smartphone and can communicate with the intersections, compute the algorithm, and show riders the timing that they can pass the intersections. Additional infrastructure can also help such as a traffic light that shows the timing to the drivers of manual vehicles.

\section{Algorithms}
\label{sec:appalgo}

\modified{Algorithm \ref{alg:downreset} shows the details in case of atomic agents. The other cases are in Appendix \ref{sec:appalgo}, but the idea is the same. The loop goes through the route in the right sequence (line \ref{dreset:loop}). Then the reset is done at all intersections after $i$ (line \ref{dreset:s}-\ref{dreset:e}).}


\begin{algorithm}[t]            
	\caption{Downstream reset (in case of atomic agents)}
	\label{alg:downreset}
	\begin{algorithmic}[1]
	\Procedure{$a'=$DownReset}{$j,a,a'$}
		\State{$\ik=j,down=False$}
    	\For{$i' \in r_k$} \label{dreset:loop}
    	\If{$down$} \label{dreset:s}
    	\State{$a'_{\tp{i'}{k}}=\emptyset$}
    	\EndIf
    	\If{$i=i'$}
    	\State{$down=True$} \label{dreset:e}
    	\EndIf
		\EndFor
	    \State{Return $a'$}
	\EndProcedure
	\end{algorithmic}
\end{algorithm}

\begin{algorithm}[t]  
	\caption{Downstream reset (in case of car agents)}
	\label{alg:downresetc}
	\begin{algorithmic}[1]
	\Procedure{$a'=$DownReset}{$j,a,a'$}
		\State{$k=j,down=False$}
    	\For{$i \in r_k$}
    	\If{$down$}
    	\State{$a'_\ik=\emptyset$}
    	\EndIf
    	\If{$a_\ik\not=a'_\ik$}
    	\State{$down=True$}
    	\EndIf
		\EndFor
	    \State{Return $a'$}
	\EndProcedure
	\end{algorithmic}
\end{algorithm}

\begin{algorithm}[t]            
	\caption{Downstream reset (in case of intersection agents)}
	\label{alg:downreseti}
	\begin{algorithmic}[1]
	\Procedure{$a'=$DownReset}{$j,a,a'$}
		\State{$i=j,down=False$}
    	\For{$k \in \{k|i\in r_k^-\}$}
    	\If{$a_\ik=a'_\ik$}
    	\State{continue}
    	\EndIf
    	\For{$i_l^k \in r_k^-$}
    	\If{$down$}
    	\State{$a'_{\tp{i_l^k}{k}}=\emptyset$}
    	\EndIf
    	\If{$i_l^k=i$}
    	\State{$down=True$}
    	\EndIf
    	\EndFor
		\EndFor
	    \State{Return $a'$}
	\EndProcedure
	\end{algorithmic}
\end{algorithm}

To address the exponential size of the search space, we replace  the action space in line \ref{brudr:space} of Algorithm \ref{alg:brudr} with the action space restricted by Algorithm \ref{alg:space}. In more detail, for each intersection on the route (line \ref{space:itri}), it checks the feasibility of each time slot (line \ref{space:itrt}), and allocates the earliest one that satisfies all the \ijcai{soft} constraints (line \ref{space:chkstart}-\ref{space:chkend}). Algorithm \ref{alg:space} shows the case of the car agents. The case of the intersection agents is in Appendix \ref{sec:appalgo}. 

\begin{algorithm}[t]	
    \caption{Fastest feasible allocation (in case of car agents)}
	\label{alg:space}
	\begin{algorithmic}[1]
	\Procedure{$a'_j=\mc{A}_j$}{$a$}
	    \State{$a'_j=a_j$}
    	\For{$i \in r_j^-$}
    	\label{space:itri}
    	\For{$t \in \mc{T}$}
    	\label{space:itrt}
    	\State{$a'_{\tp{i}{j}}=t$}
    	\label{space:chkstart}
    	\State{$feasible=True$}
    	\For{$c_l \in \mc{C}_j\setminus\{c_{delay}\}$}
    	\If{$f_l(a')>0$}
    	\State{$feasible=False$}
    	\State{$break$}
    	\EndIf
		\EndFor
		\If{$feasible$}
		\State{break}
    	\label{space:chkend}
		\EndIf
		\EndFor
		\EndFor
	    \State{Return $a'$}
	\EndProcedure
	\end{algorithmic}
\end{algorithm}

Algorithm \ref{alg:spacei} shows the $\mc{A}_j$ function for intersection agents.

\begin{algorithm}[t]	
    \caption{Fastest feasible allocation (in case of intersection agents)}
	\label{alg:spacei}
	\begin{algorithmic}[1]
	\Procedure{$a'_j=\mc{A}_j$}{$a$}
	    \State{$a'_j=a_j$}
    	\For{$k\in\{k\in\mc{K}|j\in r_k^-\}$}
    	\label{space:itri}
    	\For{$t \in \mc{T}$}
    	\label{space:itrt}
    	\State{$t_{min}=T_k^o$}
    	\If{$j\not=i_1^k$}
    	\State{$i=\prev(j,k)$}
    	\If{$a_{\tp{i}{k}}=\emptyset$}
    	\State{$a'_{\tp{j}{k}}=\emptyset$}
    	\Else
    	\State{$t_{min}=a_{\tp{i}{k}}+L_{(i,j)}$}
    	\EndIf
    	\EndIf
    	\If{$t<t_{min}$}
    	\State{continue}
    	\EndIf
    	
    	\State{$a'_{\tp{j}{k}}=t$}
    	\label{space:chkstart}
    	\State{$feasible=True$}
    	\For{$c_l \in \mc{C}_j\setminus\{c_{delay}\}$}
    	\If{$f_l(a')>0$}
    	\State{$feasible=False$}
    	\State{$break$}
    	\EndIf
		\EndFor
		\If{$feasible$}
		\State{break}
    	\label{space:chkend}
		\EndIf
		\EndFor
		\EndFor
	    \State{Return $a'$}
	\EndProcedure
	\end{algorithmic}
\end{algorithm}

\section{Communication complexity}
\label{sec:appcomm}

The communication complexity depends on what we define an agent to be.
For example, atomic agents only have to communicate with an upcoming intersection independently, while a car agent has to communicate with all intersections on its route. Similarly, intersection agents have to communicate with all cars that can pass the intersection at the same time. However, with the system design above, we do not need communication among cars, but a car only has to communicate with intersections on its route. Also, according to Lemma \ref{thm:feas}, each agent has to communicate at most once in case of atomic and car agents. The size of the message between a car and an intersection to check the available slots is at most $T*$(max degree of the road network), and then is polynomial in the problem size. When $bestU=True$ in \textproc{BRUDR}, the same message exchange as in MGM (Maximum Gain Message \cite{pearce2005local}) is needed, but it just requires one more message to share the expected improvement in utility.

\section{Negative result on unbounded PoA}
\label{sec:apppoa}

Even though BRUDR finds a pure Nash equilibrium, it is not always \modified{the maximizer of the social welfare $J(a)$, }
even in case of the car agents, where the social optimum coincides with the optimal solution of $D(a)$. 
In this section, we discuss the possibility to bound the PoA of CMMI. 
First, we define the PoA as follows.

\begin{equation}
\left.
\begin{array}{l}
PoA=\underset{\substack{{a_{ne}\in\mc{A}_{ne}}}}{\max} \frac{|J(a_{ne})|+\epsilon}{|J(a_{so})|+\epsilon}
\end{array}
\right.
\label{eq:poa}
\end{equation}

\noindent where $\mc{A}_{ne}$ is the set of all feasible Nash equilibria and $a_{so}$ is a social optimum. The lower value of PoA is desirable to guarantee the lower bound of the efficiency, and $PoA=1$ is the best.
We add a constant $\epsilon>0$ because $J(a_{so})$ can be zero. Although there are techniques to bound the PoA based on the submodularity of the objective function \cite{nemhauser1978analysis,vetta2002nash}, unfortunately they cannot be applied to CMMI. This is because our objective function $J(a)$ does not satisfy the ``non-decreasing'' property, which is a  necessary condition of using these techniques. (Allocating time slots can violate \ijcai{soft} constraints, which decreases $J(a)$). Even worse, we have the following negative result.

\begin{theorem}
The PoA of CMMI is unbounded.
\label{thm:unbound}
\end{theorem}

\begin{proof}
The following example shows a case where PoA can be unbounded.

\begin{example}
\label{thm:example}
The road network and routes of cars are shown in Figure \ref{fig:counterexp}. There are 4 cars $\mc{K}=\{1,2,3,4\}$ and 7 intersections $\mc{I}=\{1,\ldots,7\}$. The cars have different departure time as $T^o_1=T^o_2=0, T^o_3=1,T^o_4=2$. All links have the same length $L_e=5$\modified{, and $T_{UB}=1$.}

Let's assume car agents for now.
Since car 1 and car 2 depart at $t=0$, one of them has to brake at intersection $i=5$ at $t=5$. If car 1 brakes, all other cars can go through without braking, then the solution is a social optimum $a_{so}$, which is also a Nash equilibrium. In this case, the total delay is $D(a_{so})=1$ and the social welfare is $J(a_{so})=-D(a_{so})=-1$. However, if car 2 brakes at $i=5$, car 2 and 3 arrive at $i=2$ at the same time $t=11$. In this case, car 3 has to brake because $T_{UB}=1$. Also, car 4 has to brake to avoid car 3. The solution is also a Nash equilibrium, because only car 1 can improve the efficiency\modified{, but has no incentive to do so.} The total delay of the solution $a_{ne}$ is $D(a_{ne})=3$ and the social welfare is $J(a_{ne})=-3$. Thus $PoA=|J(a_{ne})+\epsilon|/|J(a_{so})+\epsilon|=(3+\epsilon)/(1+\epsilon)$. Note that we can make instances that have more cars departing from $i=4$ after car 4. In general, $PoA=(N-1+\epsilon)/(1+\epsilon)$. Then, $\lim_{N\to \infty}PoA=\infty$. Note that the Nash equilibria above are also Nash equilibria in case of other agent types. Then similarly, PoA is also unbounded in those cases.
\end{example}
\end{proof}

\modified{As we can see in Example \ref{thm:example}, ensuring smooth passage of a platoon of cars (car 3 and 4) is a key to traffic efficiency. However, Theorem \ref{thm:NPh} and \ref{thm:unbound} indicate that it is difficult to create a platoon going through multiple intersections without braking. In the following section, we propose a technique that can address this issue.}

\section{Proofs}
\label{sec:appproof}

\subsection{Proof of Theorem \ref{thm:NPh}}

\begin{proof}
We show the detail of the reduction of job shop scheduling (JSP) to CMMI. 
It is known that the following JSP is NP-hard
\cite{gonzalez1982unit}. 
There are $n\geq 1$ independent jobs $\mc{J}=\{J_1,\ldots,J_n\}$ to be processed by $m$ machines $\mc{M}=\{M_1,\ldots,M_m\}$. Each job $J_i$ consists of $l_i$ tasks, and $j$th task of $J_i$ is denoted by $\tau_{\tp{i}{j}}$. The mapping between tasks and machines is given by $q_{i,j}$, which means that task $\tau_{\tp{i}{j}}$ is processed by $M_{q_{i,j}}$. The $j+1$th task cannot be executed before the $j$th task terminates. The processing time of task $\tau_{\tp{i}{j}}$ is $t_{\tp{i}{j}}=1$. A schedule $S$ consists of a sequence of assignments of jobs to machines in such a way that:
\begin{itemize}
    \item no job meets with more than one machine at a time
    \item no machine meets with more than one job at a time, and
    \item job $J_i$ meets machine $M_j$ for exactly $t_{\tp{i}{j}}$ time units.
\end{itemize}
No task is interrupted during the process. Let $f_i(S)$ represent the finishing time for job $J_i$ in schedule $S$. The goal of the problem is to find schedule $S$ that minimises the mean finishing time (MFT), $\sum_{J_i\in\mc{J}}f_i(S)/n$.

Now we show that any of this JSP can be reduced to a special case of CMMI where:
\begin{itemize}
    \item The set of cars $\mc{K}$ is the set of jobs $\mc{J}$.
    \item The set of intersections $\mc{I}$ is the set of machines $\mc{M}$. 
    \item The route $r_k^-$ of car $k\in\mc{K}$ is the tasks of the corresponding job.
    \item All cars have identical departure time $T_k^o=0$.
    \item Once $j$th task terminates, $j+1$th task can be assigned immediately. This is represented by $L_e=0, \forall e\in\mc{E}$.
    \item The road network contains edges that cover all the consecutive two tasks.
    \item $T$ and $T_{UB}$ and $K_{UB}$ are large enough so that $h_{wb}, h_{lcap}$ and $h_{arr}$ are not violated.
    \item The objective function of CMMI $D(a)$ is MFT multiplied by constant $K$.
\end{itemize}
\end{proof}

\subsection{Proof of Lemma \ref{thm:feas}}

\begin{proof}
\begin{figure}[t]
\centering
\includegraphics[width=0.7\columnwidth]{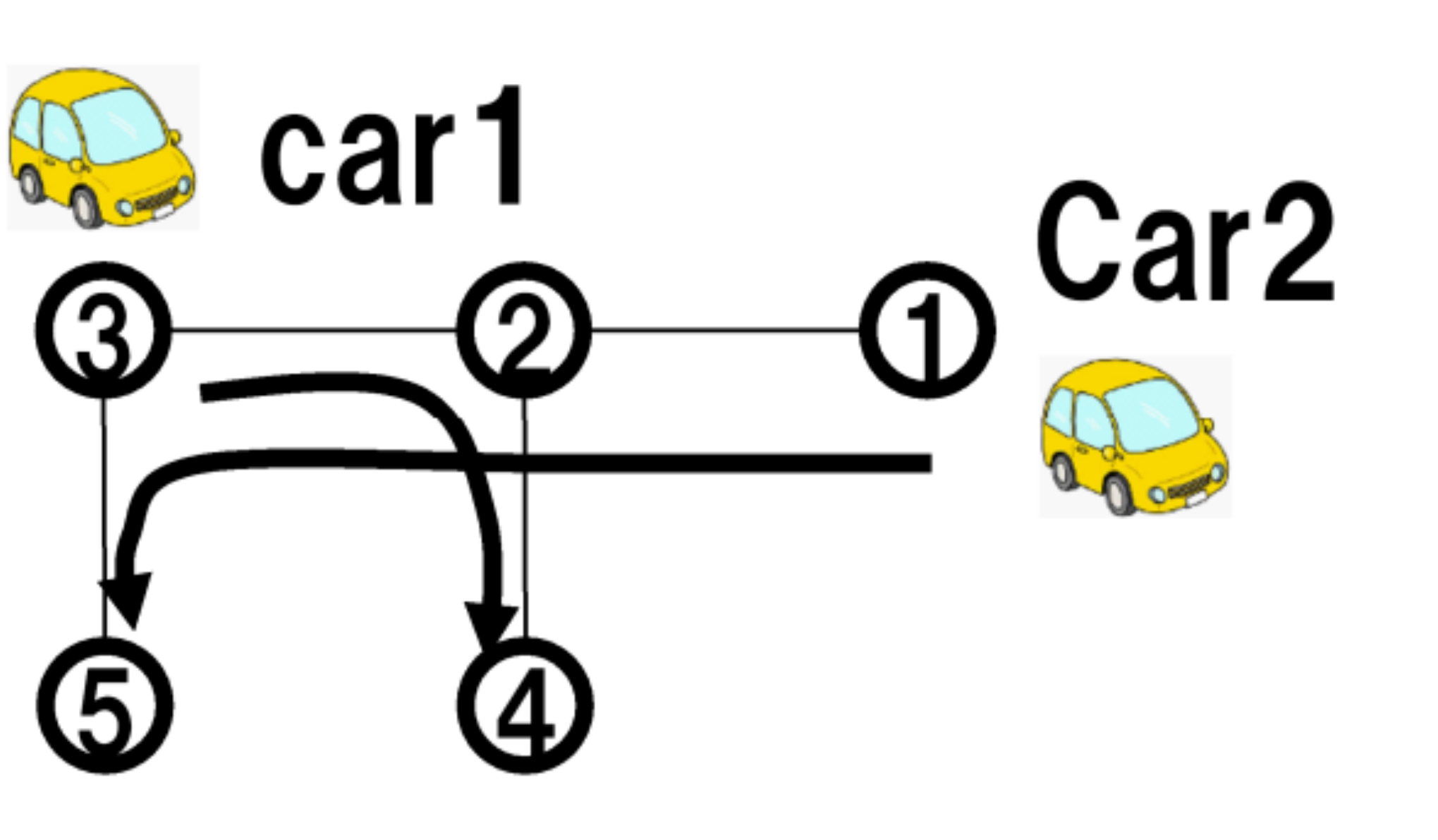}
\caption{Example where solution can be infeasible in case of intersection agents: $T^o_1=T^o_2=0$ and $L_e=5$.}
\label{fig:infeasible}
\end{figure}

First, we assume car agents, or atomic agents.
Because of the order of the penalty constants as in (\ref{eq:penalty}), BRUDR tries to satisfy the constraints one by one with the priority. First, the largest $P_{fwd}$ forces updates from upstream agents, and the first term of $c_{fwd}$ is never violated by the downstream reset. Next, agents can always find the earliest allocation that satisfy $c_{col}$ and $c_{lexc}$, because agents can delay their trips as much as they need to avoid collisions without the 3 bounding constraints, $c_{wb}, c_{lcap}$ and $c_{arr}$. Then, the other hard constraints, the second term of $c_{fwd}$ and $c_{dep}$ also can be satisfied with the same reason without violating other constraints. Since agents can find the earliest allocation without violating the hard constraints of other agents, and later updating agents do not change the earliest allocations of already allocated agents, BRUDR converges to a pure Nash equilibrium with each agent updating the action at most once.

This does not hold in case of intersection agents, and Figure \ref{fig:infeasible} shows an example. The two cars depart at the same time, and conflict at intersection 2. Let's say BRUDR allocates from the smallest ID. Then, agent of intersection 1 allocates the departure time of car 2 first, then intersection 2 computes the allocation. Because of the largest $P_{fwd}$, car 1 cannot be allocated yet and then only car 2 is allocated, as $a_{2,1}=\emptyset,a_{2,2}=5$ for example. Then intersection 3 allocates both cars. After the first updates of all intersections, intersection 2 tries to update again because car 1 is not allocated there yet. Then it tries to update as $a_{2,1}=5,a_{2,2}=6$, for example. However, this fails because car 2 also updates at the intersection and has to reset at the downstream intersection 3, which results in the same $P_{none}$ but increases a unit time of delay. Then the update is not the best response and is not executed, remaining $a_{2,1}=\emptyset$ which is infeasible.
\end{proof}

\section{Experimental settings}
\label{sec:appexp}

\begin{itemize}
    \item We call numpy.random.seed at the beginning of the main code with different seed for each run of simulation. 
    \item We use CPLEX 12.7.1, Python 3.8, Red Hat Enterprise Linux Server release 6.10 and Intel Xeon CPU E5-2670 (2.60 GHz) with 8 cores, 132 GB memory to run the experiments.
    \item The source code is in the supplementary material. \camera{We do not use any DCOP simulators. All the codes are implemented from scratch.}
    \item We use pNEUMA dataset of zone 4, on the date of 24/10/2018 from 8:30 to 9:00 (Figure \ref{fig:epfl}).
    \item The real-world data of vehicle trajectory is provided in Latitude-Longitude coordinate system. The code converts the data to East-North-Up coordinate. Then the code scans all data points in a vehicle trajectory and decides that an intersection is on the route if a data point is inside a radius ($r=5$ [m]) from the center of the intersection. This conversion procedure is also included in the source code.
\end{itemize}

\begin{figure}[t]
\centering
\includegraphics[width=1.0\columnwidth]{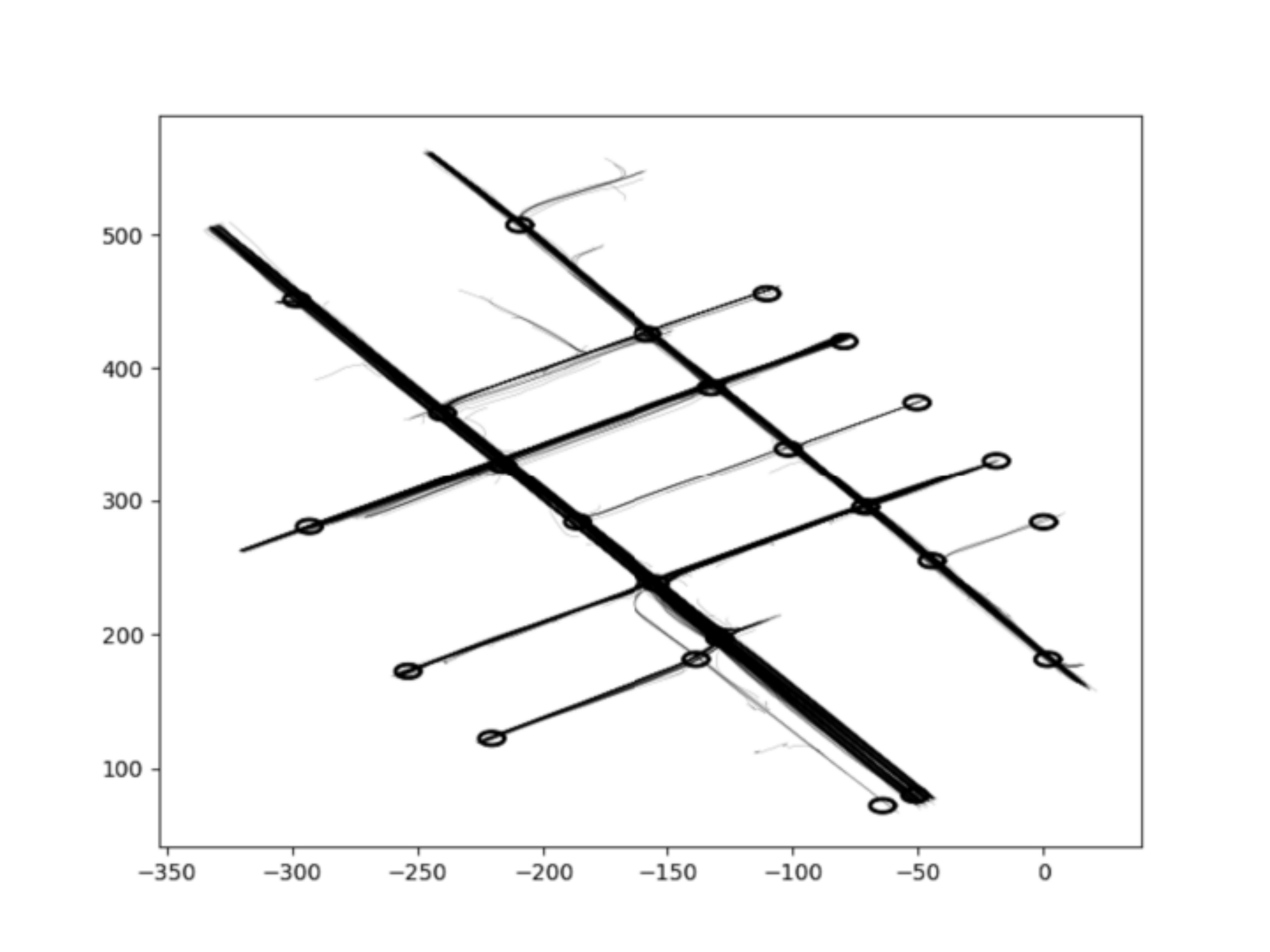}
\caption{Vehicle trajectories in Athens. The circles are intersections.}
\label{fig:epfl}
\end{figure}

\section{Experiment on externality among vehicles}
\label{sec:appextern}

As we saw in Figure \ref{fig:smalleff}, the difference between algorithms disappears when the problem size becomes large. However, cars can have negative externalities that cause cascades of braking as in Figure \ref{fig:counterexp}, even in large problems. To see this aspect, we compute the negative externality of agent $i$ as follows.

\begin{equation}
\left.
\begin{array}{l}
ex_i=D(a_{ne})-D(a_{ne}^{-i}),
\end{array}
\right.
\end{equation}

where \modified{$a_{ne}^{-i}$} is a Nash equilibrium when agent $i$ is absent. Note that $ex_i$ is positive and a larger value is worse. Figure \ref{fig:extern} shows the worst externality of agents $\max_i ex_i$ with different number of cars $K$, when using $default$. For each $K$, we conducted 200 experiments. Also note that the computational complexity of $ex_i$ is large and this experiment is difficult to scale. It shows that the negative externality of agents becomes larger when the size of the problems becomes larger. Meanwhile, we also see the improvement by $pba$ defined as $D(a_{ne}^{default})-D(a_{ne}^{pba})$, where $a_{ne}^{default}$ and $a_{ne}^{pba}$ are Nash equilibria computed by the default BRUDR and BRUDR starting from $pba$, respectively. Figure \ref{fig:improve} shows the ratio of the improvement over the worst negative externality. It shows that the improvement by $pba$ becomes larger when $K$ increases, and reaches about 8\% of the worst externality.

\begin{figure}[t]
\centering
\includegraphics[width=0.85\columnwidth]{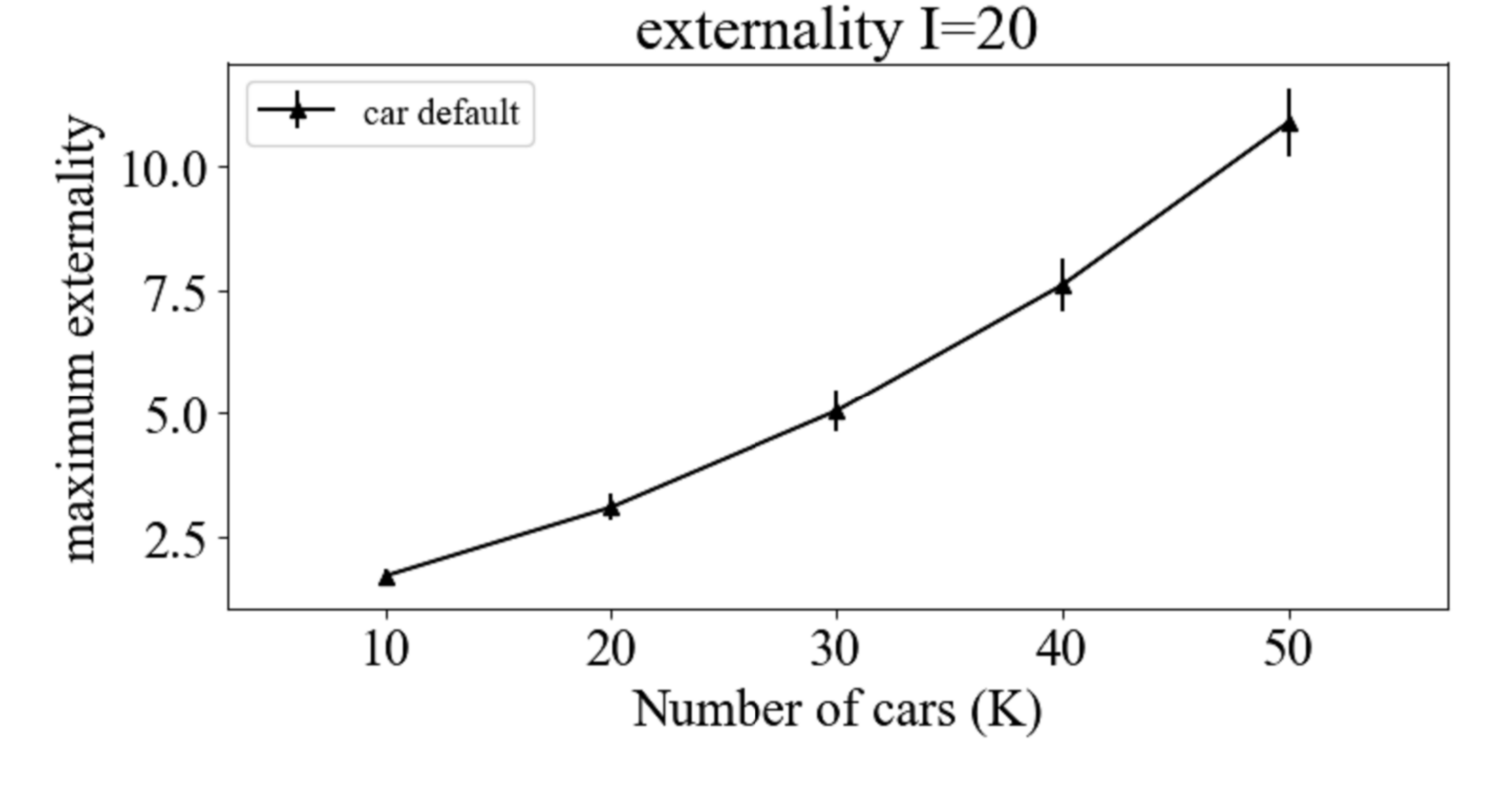}
\caption{The worst externality of agents $\max_i ex_i$ with different number of cars $K$ when cars are controlled by $default$.}
\label{fig:extern}
\end{figure}

\begin{figure}[t]
\centering
\includegraphics[width=0.85\columnwidth]{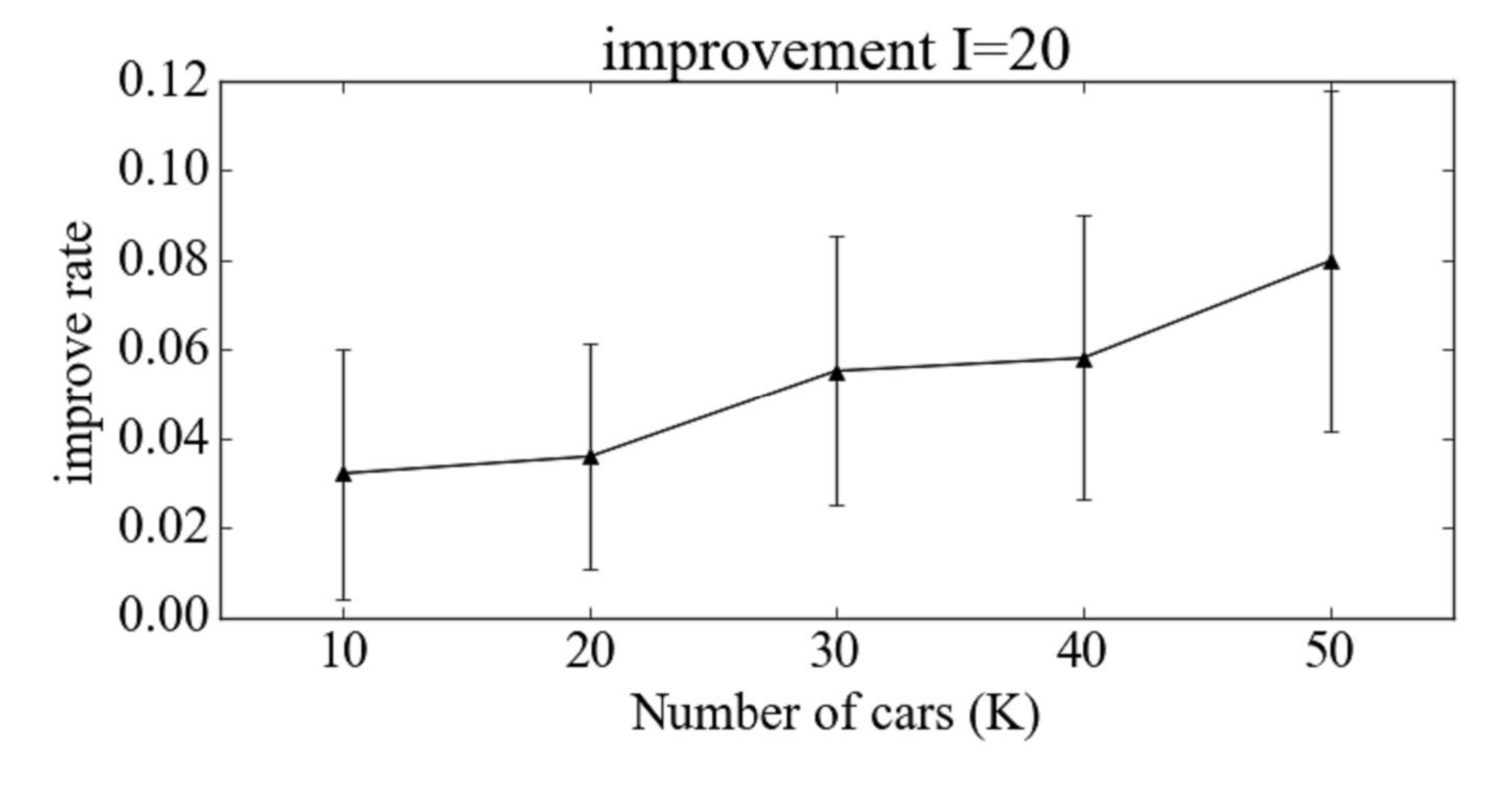}
\caption{The ratio of improvement in externality by $pba$ with different number of cars $K$.}
\label{fig:improve}
\end{figure}

\color{\newcolor}
\section{Reproducibility Checklist}
\label{sec:apprepro}

\subsection{Algorithm}

If the paper introduces a new algorithm, it must include a conceptual outline and/or pseudocode of the algorithm for the paper to be classified as CONVINCING or CREDIBLE. (CONVINCING)

\subsection{Theoretical contribution}

If the paper makes a theoretical contribution: 
\begin{enumerate}
    \item All assumptions and restrictions are stated clearly and formally (yes)
    \item All novel claims are stated formally (e.g., in theorem statements) (yes)
    \item Appropriate citations to theoretical tools used are given (yes)
    \item Proof sketches or intuitions are given for complex and/or novel results (yes)
    \item Proofs of all novel claims are included (yes)
\end{enumerate}

For a paper to be classified as CREDIBLE or better, we expect that at least 1. and 2. can be answered affirmatively, for CONVINCING, all 5 should be answered with YES. (CONVINCING)

\subsection{Data sets}

If the paper relies on one or more data sets: 

\begin{enumerate}
    \item All novel datasets introduced in this paper are included in a data appendix (NA)
    \item All novel datasets introduced in this paper will be made publicly available upon publication of the paper (NA)
    \item All datasets drawn from the existing literature (potentially including authors’ own previously published work) are accompanied by appropriate citations (yes)
    \item All datasets drawn from the existing literature (potentially including authors’ own previously published work) are publicly available (yes)
    \item All datasets that are not publicly available (especially proprietary datasets) are described in detail (NA)
\end{enumerate}

Papers can be qualified as CREDIBLE if at least 3., 4,. and 5,. can be answered affirmatively, CONVINCING if all points can be answered with YES. (CONVINCING)

\subsection{Experiments}

If the paper includes computational experiments:

\begin{enumerate}
    \item All code required for conducting experiments is included in a code appendix (yes)
    \item All code required for conducting experiments will be made publicly available upon publication of the paper (yes)
    \item Some code required for conducting experiments cannot be made available because of reasons reported in the paper or the appendix (NA)
    \item This paper states the number and range of values tried per (hyper-)parameter during development of the paper, along with the criterion used for selecting the final parameter setting (yes)
    \item This paper lists all final (hyper-)parameters used for each model/algorithm in the experiments reported in the paper (yes)
    \item In the case of run-time critical experiments, the paper clearly describes the computing infrastructure in which they have been obtained (yes, in Appendix \ref{sec:appexp})
\end{enumerate}

For CREDIBLE reproducibility, we expect that sufficient details about the experimental setup are provided, so that the experiments can be repeated provided algorithm and data availability (3., 5., 6.), for CONVINCING reproducibility, we also expect that not only the final results but also the experimental environment in which these results have been obtained is accessible (1., 2., 4.). (CONVINCING)
\color{black}


\fi


\end{document}